\documentclass[preprint, 12pt]{elsarticle}

\usepackage[utf8]{inputenc}
\usepackage[T1]{fontenc}
\usepackage{graphicx}
\usepackage{amssymb,amsthm,mathtools}
\usepackage{microtype}
\usepackage{bm}
\usepackage{booktabs} 
\usepackage{siunitx}
\usepackage{xcolor} 
\usepackage{hyperref}
\usepackage[margin=2.5cm]{geometry} 

\newtheorem{theorem}{Theorem}[section]

\newtheorem{remark}[theorem]{Remark}

\newcommand{\R}{\mathbb{R}}

\newcommand{\beq}{\begin{equation}}
\newcommand{\eeq}{\end{equation}}
\newcommand{\order}[1]{\mathcal{O}(#1)}
\newcommand{\supp}{\text{supp}}

\journal{Chaos, Solitons \& Fractals}

\begin{document}

\begin{frontmatter}

\title{Uniform Asymptotics of the Pseudo Wigner–Ville Distribution for Nonlinear Chirps}


\author[aff1,aff2]{Vincenzo Pierro\corref{cor1}}
\ead{pierro@unisannio.it}
\cortext[cor1]{Corresponding author}

\author[aff3]{Maurizio Feo}

\author[aff4]{Maurizio Ricciardi}

\author[aff5]{Rocco P. Croce}

\author[]{Theo Demma\fnref{fn1}}
\fntext[fn1]{Deceased}

\author[aff2]{Innocenzo M. Pinto}

\author[aff6]{Paolo Addesso}

\address[aff1]{Dipartimento di Ingegneria, Università del Sannio, Corso Garibaldi 107, I-82100 Benevento, Italy}

\address[aff2]{INFN, Sezione di Napoli, Gruppo Collegato di Salerno, I-84084 Fisciano, Salerno, Italy}

\address[aff3]{Liceo Statale "Alfano I", I-84132 Salerno, Italy}

\address[aff4]{European Patent Office, 80335 Munich, Germany}

\address[aff5]{Dipartimento di Ingegneria Industriale, Università degli Studi di Salerno, I-84084, Fisciano (SA), Italy}

\address[aff6]{Dipartimento di Ingegneria dell’Informazione, Elettrica e Matematica Applicata, Università degli Studi di Salerno, I-84084, Fisciano (SA), Italy}


\begin{abstract}
The analysis of non-stationary signals in complex physical systems often relies on Time-Frequency distributions. Among these, the Pseudo Wigner–Ville Distribution (PWVD) stands out for its superior resolution but is mathematically challenging due to its inherent quadratic nonlinearity.
 This nonlinearity generates complex interference artifacts and cross-terms in the phase space, 
potentially obscuring the physical features of the signal, particularly for nonlinear chirps.
In this work, we establish { a mathematically grounded framework for the PWVD for general windowed nonlinear chirps}. By leveraging the theory of oscillatory integrals with coalescing stationary points, we derive a uniform asymptotic expansion that bridges the gap between heuristic signal processing and semiclassical geometric approaches (Berry’s chord construction). The resulting closed-form representation, expressed in terms of symmetric incomplete Airy functions, provides a unified description of the nonlinear transform's behavior, regularizing the transition across the instantaneous frequency caustics.
While the framework is general, { we show its power on two illustrative examples: the high‑precision nonlinear chirps of coalescing binaries in gravitational‑wave astronomy} and 
 radar nonlinear chirps for pulse compression applications.
The analytical results successfully predict the structure of interference patterns and quantify the systematic bias in peak-based frequency estimation. 
Therefore, {this study establishes a systematic bridge} between nonlinear mathematical analysis and precision experimental physics, validating the PWVD as a robust tool for detailed source characterization in high-signal-to-noise regimes.
\end{abstract}


\end{frontmatter}

\section{Introduction}
The study of nonlinear dynamics in complex systems—ranging from fluid turbulence to biological rhythms and astrophysical phenomena—requires robust tools to characterize non-stationary signals whose frequency content evolves over time ~\cite{KantzSchreiber, Strogatz}. 
Time-Frequency (TF) analysis is considered one of the most suitable frameworks to face this problem~\cite{FlandrinBook1998, FlandrinBook2018, boashashbook, Seo2025, Warion2024, Ugarte2023} and the Wigner-Ville Distribution (WVD) and its windowed counterpart, the Pseudo Wigner-Ville Distribution (PWVD)~\cite{Wigner1932, Ville1948}, represent a cornerstone in this field. 
Originally conceived in quantum mechanics to describe phase-space probabilities, the Wigner distribution is a bilinear (quadratic) transform that offers the highest possible time-frequency resolution~\cite{Cohen1995}. 
However, its application to nonlinear signals is non-trivial: the very nonlinearity of the transform that grants its resolution also introduces intricate interference structures (cross-terms) and oscillatory artifacts~\cite{Mecklen1, Mecklen2, TFTB}.

To mitigate these inherent cross-terms and artifacts, often at the expense of a sharp TF representation, other transforms, often referred to as \emph{smoothed} WVD, 
have been introduced, identifying the wide Cohen's class distribution \cite{FlandrinBook2018}, which also includes the WVD as a special case. 
Moreover, modern and sharper techniques—such as Time-Frequency reassignment and the Synchrosqueezing Transform (SST) \cite{Daubechies2011, Auger2013, 
FlandrinBook2018} -- are frequently employed too. 
These methods do not constitute standalone analytical transforms; 
rather, they are complex algorithmic procedures designed to reallocate the signal's energy across the phase space (it is worth noting that SST itself operates essentially as a highly advanced, phase-preserving variation of reassignment). 
In this paper, our goal is to provide a fundamental, closed-form mathematical description of the underlying unperturbed PWVD. 
We believe that our uniform asymptotic framework can serve as the theoretical foundation necessary to justify these algorithms.

In the context of nonlinear science, the aforementioned artifacts are not merely {\it noise} to be discarded,
 but deterministically generated features arising from the geometry of the signal in the phase space. 

Interpreting these features requires a mathematical framework capable of handling the singularities that emerge when the transform is applied to 
nonlinear chirps~\cite{chassande1999, Boashash1992part1}. 
This challenge is interdisciplinary by nature, connecting the asymptotic analysis of diffraction integrals in optics, the semiclassical approximations in quantum mechanics (notably Berry's geometric chord approach~\cite{Berry1977}), and modern signal processing.

In fact, the asymptotic analysis of Wigner distributions has a rich history across different physics and engineering domains. 
In standard signal processing, classical Stationary Phase Methods (SPM)~\cite{Erdelyi} have been extensively used to characterize the time-frequency ridges of chirp signals~\cite{Stankovic2001, Flandrin1999_SPM}. 
However, these classical approximations inherently diverge at the inflection points of the frequency track (caustics), where the second derivative of the phase vanishes.
Concurrently, in semiclassical quantum mechanics, the WKB~\cite{BenderOrszag} limit of the Wigner function has been profoundly analyzed using catastrophe theory 
\cite{Balazs1973, Berry1977}. While these foundational works successfully describe the infinite-domain behavior of phase-space distributions, 
developing a \textit{uniform} asymptotic theory specifically tailored to the windowed PWVD—which must mathematically account for finite boundary
truncation effects across the caustics—has remained an open challenge.

The refinement of these nonlinear analysis tools is key for high precision estimation of instantaneous frequency, which is of paramount importance in several fields of science and technology. Some examples are listed below.
\begin{itemize}

\item Biomedical signals, which are inherently non-stationary and require high precision TF methods in order to improve the diagnostic effectiveness~\cite{Wacker2013, boashashbook}. Indeed, they are essential in neurology for the analysis of electroencephalograms and magnetoencephalographies  in order to predict and/or identify abnormalities~\cite{Boashash2018}, such as epileptic seizures, or to properly design brain-computer interfaces~\cite{PachoriBook2023}.  They are also widely used in cardiorespiratory studies~\cite{Wacker2013, Tirtom2008}, in order to allow scientists to quantify both brief and sustained changes in heart rate, breathing, and blood pressure, and (in general) they are key in several diagnostic techniques, ranging from breast ultrasound image classification~\cite{Xiong2023} to the monitoring of the surface electromyography for muscle and neuromuscular evaluations in the prevention and rehabilitation fields~\cite{Muceli2024}.

\item Mechanical fault diagnosis and monitoring is a critical task. Indeed, industrial operations rely heavily on moving components like shafts, spindles, planetary gearboxes, rolling bearings and flywheels, and a failure in any single part can lead to expensive downtime and significant financial losses~\cite{PachoriBook2023, Li2025, Gelman2007,Jiang2023,Zhao2024}. 
Therefore, in order to properly detect and classify these issues, the machine's vibration signals, that are strongly non-stationary, are analyzed by using high precision TF techniques.

\item Gravitational waves (GW) emitted by coalescing binaries are pristine examples of nonlinear chirps, where the instantaneous frequency evolves rapidly over time. Following the groundbreaking discoveries by Virgo~\cite{Virgosite}, LIGO~\cite{LIGOsite}, and KAGRA~\cite{KAGRA}, and with the prospect of third-generation observatories like the Einstein Telescope~\cite{ET} and Cosmic Explorer~\cite{CExplore}, we are entering an era of high Signal-to-Noise Ratio (SNR) detections,
as exemplified by the recent observation of GW250114~\cite{GW250114}. 
In this high SNR regime ($\sim 80$)  precision tests of fundamental physics become possible and subtle physical effects can be identified 
even in the absence of fully accurate waveform models~\cite{Bohe2017,Pratten2021,Varma2019}.

\item Other relevant fields are related to the analysis of seismic signals~\cite{WangBook2023}, the analysis of nonlinear and non-stationary systems in fluid dynamics
 (e.g. ocean and rogue waves~\cite{Osborne1995,Hwang2003, Tikan2022} or climate analysis~\cite{MA2015}), the study of nonlinear optical systems~\cite{Mandal2004},
  nonlinear dynamics in plasma physics~\cite{Wang2022,Wang2024}, { nonlinear radar chirps for pulse-compression applications \cite{Roy2021NLC}} 
  and many other fields: the interested reader can find an (almost) complete list within the Flandrin's book~\cite{FlandrinBook2018}.

\end{itemize}

Popular TF methods, such as the wavelet-like Q-transform~\cite{QtransformRef}, avoid interference but are not capable to achieve the high precision of the PWVD, which is required for future scientific and technological objectives. 
Indeed, its resolution, like that of the aforementioned TF representation, is bounded by the uncertainty principle; 
however, the PWVD can achieve a resolution closer to the Cram\'er-Rao bound~\cite{Boashash1992part2, Boashash1992_XWVD, Boashash1990_TFA}, provided that its nonlinear 
artifacts are analytically tamed.

In this work, we derive a uniform asymptotic expression for the PWVD of general windowed nonlinear chirps. 
Our approach is grounded in the theory of oscillatory integrals with two coalescing stationary points, pioneered by Chester, Friedman, and Ursell~\cite{chester1957}, and refined by Bleinstein~\cite{Bleistein1966}. 

While the foundational mathematical machinery was present in applied mathematics, 
its exact translation to the finite-duration nonlinear chirps in Time-Frequency analysis is genuinely novel. 
Our core contribution lies in adapting this classical theory to the specific topological structure of the windowed PWVD. 
By explicitly mapping the finite integration limits to symmetric incomplete Airy functions, we provide a mathematical justification 
for the amplitude pre-factors and boundary truncation effects that are absent in purely geometric or infinite-domain approaches.

We show that the PWVD behavior is universally governed by incomplete Airy functions~\cite{Abramowitz1964}, providing a physical explanation for the {\it ringing} artifacts observed in numerical analyses. 
This framework unifies three perspectives: the heuristic signal processing view, the rigorous asymptotic analysis~\cite{Erdelyi}, and the semiclassical geometric interpretation. 
Specifically, we demonstrate that in the short-chord limit, our uniform expansion converges to the local heuristic formula and aligns 
with Berry's semiclassical chord construction~\cite{Berry1977}, moving the heuristic approach from a convenient simplification to a mathematically justified result.

From a practical standpoint, this work yields a powerful analytical tool for high-precision signal analysis. 
We apply our asymptotic formula to derive an analytical expression for the systematic bias inherent in peak-based instantaneous frequency estimation, linking the bias directly to the local curvature of the signal's time-frequency track. 
To validate this interdisciplinary framework, we consider a zero-order post-Newtonian model~\cite{Sathya} of a whitened GW signal. 
The excellent agreement between our asymptotic formula and direct numerical computations confirms that a deep mathematical understanding of the nonlinear transform allows for the precise removal of systematic biases.

The paper is structured as follows. After establishing the preliminaries in Sec.~\ref{sec:preliminaries} and the nonlinear chirp model in Sec.~\ref{sec:modeling}, we build intuition with a heuristic derivation in Sec.~\ref{sec:heuri}. 
Section~\ref{sec:semiclassical} places this result within Berry's semiclassical geometric framework. 
Section~\ref{sec:uniform} presents the core contribution: the derivation of the uniform asymptotic formula. 
The practical implications are explored in Sec.~\ref{sec:bias}\,\, i.e. the systematic bias quantification, Sec.~\ref{sec:GWchirp}\,\, i.e. the numerical validation on GW chirps, { and Sec.~\ref{sec:ConcaveChirp}\,\, i.e. the numerical validation on radar nonlinear chirps.}
Finally, Sec.~\ref{sec:conclusions} summarizes our findings.
\section{Definitions and Preliminaries}
\label{sec:preliminaries}

In this section we introduce the cross pseudo-Wigner--Ville distribution (X-PWVD), its analytic-signal formulation, and the suppression of spurious components 
(see \cite{boashashbook} for a comprehensive reference).

Given two complex-valued signals $x(t)$ and $y(t)$, the X-PWVD is defined as
\begin{equation}
PW_{xy}(t,f) =
\int_{-T}^{T}
w\left(\frac{\tau}{2}\right)^2
x\left(t+\frac{\tau}{2}\right)
y^{\ast} \left(t-\frac{\tau}{2}\right)
e^{-\mathrm{i} 2\pi f \tau}\, d\tau,
\label{eq:xpwvd}
\end{equation}
where \( w(\tau) \) is a real-valued symmetric window function defined in the interval $[-T/2, T/2]$ and \( (\cdot)^{*} \) denotes complex conjugation. The choice of \( w(\tau) \) controls the trade-off between 
time and frequency resolution, mitigating cross-term artifacts in the distribution. 
The X-PWVD reduces to the standard PWVD when $x(t) = y(t)$. 
Furthermore, when the windowing function is absent and the integration is performed 
over an infinite interval, this simplifies to the classic 
Wigner-Ville distribution (WVD) used in quantum mechanics \cite{Wigner1932}.

Throughout this work we use the Fourier transform convention
\begin{equation}
\mathcal{F}\{g(t)\}(f) = \int_{-\infty}^{+\infty} g(t) e^{-\mathrm{i} 2\pi f t} \, dt,
\end{equation}
and adopt the shorthand notation \( \hat{g}(f) = \mathcal{F}\{g(t)\}(f) \). Time is denoted by \( t \), frequency by \( f \), and lag variables by \( \tau \).

For real-valued signals, it is convenient to introduce the analytic signal
\begin{equation}
s_a(t) = s(t) + \mathrm{i} \mathcal{H}\{s(t)\},
\label{eq:analytic}
\end{equation}
where $\mathcal{H}\{\cdot\}$ is the Hilbert transform. In the frequency domain, its Fourier transform can be written as
\begin{equation}
\mathcal{F}\{s_a(t)\}(f) = 2u(f)\,\hat{s}(f),
\label{eq:analytic-spectrum}
\end{equation}
with $u(f)$ the Heaviside step function, showing that only positive frequencies are retained.

Using this representation, the decomposition of the X-PWVD into positive- and negative-frequency contributions simplifies. 
In particular, only the correlation between positive-frequency components is physically meaningful. When analytic signals $x_a(t)$ and $y_a(t)$ are used, 
all spurious terms vanish, and the distribution reduces to
\begin{equation}
PW_{x_a,y_a}(t,f) = 4\, PW_{x^+,y^+}(t,f),
\label{eq:analytic-xpwvd}
\end{equation}
where  \( x^{+}(t) \) and \( y^{+}(t) \) contain only positive frequencies of real signal. 
The mirror contribution arising from the terms $x^{-}(t)$ and $y^{-}(t)$ with respect to the frequency origin $f=0$, as well as the cross-interference terms between opposite frequency components, are suppressed.
Equation (\ref{eq:analytic-xpwvd}) provides  a clean representation of the relevant correlation.

To summarize, the analytic-signal formulation of the X-PWVD suppresses redundant mirror terms and cross-interference components that arise in the real-signal case. 
This yields a time-frequency distribution focused exclusively on positive-frequency 
interactions, simplifying interpretation and enhancing physical relevance, which motivates 
the wide use of the X-PWVD of the analytic signal in applications  \cite{boashashbook, Cohen1995}.

\section{Modeling Chirp Signals for Time-Frequency Analysis}
\label{sec:modeling}

This work focuses on the analysis of nonlinear chirp signals, which are characterized by an instantaneous frequency that varies over time.
A key challenge is to establish a well-defined model for such signals, particularly their representation as an analytic signal from 
which instantaneous amplitude and phase can be extracted. We explore two complementary approaches: 
a constructive method based on asymptotic expansions and a foundational analysis based on spectral conditions.

\subsection{Asymptotic Chirps: A Constructive Approach}
In many signal processing applications, a key class of signals are \emph{chirps}, which are characterized by a time-varying instantaneous frequency. Mathematically, these are AM-FM signals of the form
\begin{equation}
    s(t) = a(t) \cos(\phi(t)),
\end{equation}
where $a(t)$ and $\phi(t)$ are the instantaneous amplitude and phase, moreover the instantaneous frequency is given by $f_{i}(t) = \frac{1}{2\pi}\phi'(t)$.
The Wigner-Ville analysis of such signals is greatly simplified by using the analytic signal representation, which isolates the positive frequency component and provides a direct link to the signal's phase and amplitude.

The analytic signal is defined as $s_a(t) = s(t) + \mathrm{i}\mathcal{H}\{s(t)\}$, where $\mathcal{H}\{s(t)\}$ is the Hilbert transform of the real signal $s(t)$. The Hilbert transform in the time-domain is defined by the principal value integral:
\begin{equation}
\mathcal{H}\{s(t)\} = \frac{1}{\pi} \operatorname{p.v.} \int_{-\infty}^{\infty} \frac{s(\tau)}{t-\tau}\,d\tau.
\label{eq:hilbtime}
\end{equation}
For a general chirp, this integral (\ref{eq:hilbtime}) is often intractable to compute exactly. However, for signals where the phase $\phi(t)$ is rapidly varying, asymptotic methods provide a powerful and accurate approximation.
When the phase contains a large parameter $N$, i.e., $s_N(t) = a(t)\cos(N\phi(t))$, we can derive a systematic expansion for its Hilbert transform.

A theorem, stated and detailed in \ref{app:hilbert_theorem}, provides the asymptotic formula for the analytic signal of a chirp. This theorem establishes that, under minimal smoothness assumptions on the amplitude $a(t)$ 
and phase $\phi(t)$, the analytic signal can be accurately approximated. Using Theorem~\ref{th:asinto_appendix}, we can write the analytic signal corresponding to $s_N(t)$ as:
\begin{equation}
    s_{a,N}(t) \approx a(t)e^{\mathrm{i} N\phi(t)} - \frac{\mathrm{i}}{N} \frac{d}{dt}\left(\frac{a(t)}{\phi'(t)}\right) e^{\mathrm{i}N\phi(t)}.
    \label{eq:asinhilbert}
\end{equation}
This approximation is the foundation for the asymptotic analysis of the X-PWVD for chirp signals.
It is important to highlight the physical interpretation of this result. The leading term, $a(t)\exp(\mathrm{i}N\phi(t))$, represents the elementary approximation, which is obtained under the {\it adiabatic hypothesis}. 
This hypothesis assumes that the signal's parameters (amplitude and frequency) change so slowly that the signal behaves locally like a pure sinusoid. 
The next term in the expansion, of order $\mathcal{O}(1/N)$, provides the first-order correction to this idealization. It precisely quantifies the deviation from the adiabatic behavior, which is caused by the non-zero rate of change of the amplitude $a(t)$ and frequency $\phi'(t)$. 
Thus, the expansion offers a direct way to measure the accuracy of the elementary approximation.

It should be noted that the asymptotic correction in Eq.~(\ref{eq:asinhilbert}) is not optional but necessarily present. 
Indeed, as emphasized in the paper \cite{picinbono1997}, the equality between the analytic representation of the signal and the leading terms of (\ref{eq:asinhilbert}) 
cannot hold in general, except for special cases. In practice, the leading term often referred to as the 
analytic signal is only an asymptotic approximation, which is valid in a rigorous sense only when an exact 
frequency separation occurs. This limitation is formally expressed by the Bedrosian theorem \cite{Bedrosian1962, Chiollaz1978}, which provides 
the precise conditions under which the analytic signal construction is exact.

A notable property of Eq.~(\ref{eq:asinhilbert}) is that its asymptotic correction acts 
exclusively on the signal's amplitude. Consequently, the instantaneous frequency---a key 
physical attribute---remains entirely unaltered.

The asymptotic approximation (\ref{eq:asinhilbert}) relies on the assumption that the phase $\phi(t)$ is sufficiently smooth.
In the following subsection, we examine the fundamental conditions that a signal must satisfy for this assumption to be justified.

\subsection{Fundamental Conditions for the Amplitude/Phase Model}
\label{sec:phase_analyticity}

While the asymptotic model provides a practical way to construct an analytic signal with a well-behaved phase function, a more fundamental question concerns the possibility of representing 
a signal $s(t)$ as $a(t)\cos(\phi(t))$ with a real-analytic phase.
This subsection provides the theoretical foundation and justification for the assumption underlying the chirp models employed throughout this paper. 

Specifically, it examines the conditions under which a real signal $s(t)$ admits such a representation while ensuring that the phase $\phi(t)$ is real analytic. 
The analysis is closely related to the theory of analytic signals \cite{Bedrosian1962, Chiollaz1978}.

Let $s(t)$ be a real signal and $s_a(t) = s(t) + \mathrm{i} \mathcal{H}\{s(t)\}$ its corresponding analytic signal. 
The phase of $s_a(t)$ is formally $\theta(t) = \text{Im}[\log(s_a(t))]$. For this phase to be extendable to an analytic function $\theta(\zeta)$ 
in a neighborhood $U \subset \mathbb{C}$, two fundamental conditions on the analytic signal's extension, $Z(\zeta)$, must be met.

\begin{theorem}[Fundamental Conditions for Phase Analyticity]
For the phase function $\theta(\zeta)$ to be analytic in a connected open neighborhood $U$, two conditions are required:
\begin{enumerate}
    \item \textbf{Analyticity of the Signal:} The signal $s(t)$ must be real analytic, allowing its extension to an analytic function $Z(\zeta)$ in $U$.
    \item \textbf{Non-Vanishing Analytic Signal:} The analytic signal extension $Z(\zeta)$ must have no zeros within the neighborhood $U$, i.e., $Z(\zeta) \neq 0$ for all $\zeta \in U$.
\end{enumerate}
\end{theorem}

\begin{proof}
The first condition is foundational. 
The second arises because the complex logarithm $\log(w)$ has a branch point singularity at $w=0$. 
If $Z(\zeta_0) = 0$ for some $\zeta_0 \in U$, then $\log(Z(\zeta_0))$ is undefined, inducing a singularity in the phase.
\end{proof}

The requirement that the real-valued signal $s(t)$ must be \textit{real-analytic} might initially appear to be an overly restrictive condition for signals encountered in practical signal processing.
Real-world signals are often time-limited or may exhibit non-smooth behavior, which would violate the definition of a real-analytic function.

However, this assumption becomes highly justifiable when viewed within the context of typical signal processing workflows. 
Signals are almost invariably subjected to pre-processing steps, such as anti-aliasing or band-selection filtering, before any advanced analysis is performed.
Such filtering operations inherently regularize the signal, making the analyticity assumption a very reasonable and powerful mathematical model. 
The justification can be more formally established by considering two prevalent signal classes:

\begin{itemize}
    \item \textbf{Band-Limited Signals:} An ideal low-pass filter transforms a signal into a band-limited one. According to the \textbf{Paley-Wiener theorem},
     a function is band-limited if and only if it is the restriction to the real line of an entire function of exponential type. An entire function is analytic over the whole complex plane. 
     Therefore, for any signal that can be modeled as perfectly band-limited, the assumption of analyticity is not an approximation but a direct mathematical consequence.

    \item \textbf{Essentially Band-Limited Signals:} In practice, perfect band-limitation is not achievable. However, the spectra of many physically relevant signals decay very rapidly (e.g., exponentially) 
    outside of a certain frequency band. Such signals are aptly modeled by function classes like the \textbf{Hardy spaces} ($H^p$). 
    A key property of functions in these spaces is that they are analytic in a strip surrounding the real axis in the complex plane. 
    This local analyticity is entirely sufficient for the uniform asymptotic methods and analytic continuation arguments employed in the paper's theoretical framework.
\end{itemize}

In conclusion, while the condition of real-analyticity on $s(t)$ seems strong in isolation, it is a pragmatic and well-founded premise for the signals that result from standard filtering operations. 
These operations smoothen the signal to a degree where its mathematical properties converge to those of analytic functions, 
thereby validating the theoretical foundation upon which the paper's analysis is built.

Consequently, this analyticity requirement perfectly accommodates practical signals acquired from biomedical sensors (e.g., EEG, ECG) or mechanical vibration monitors. 
In these applied fields, raw analog signals are subjected to anti-aliasing hardware filters prior to digitization. 
The translation of these physical acquisition protocols into the mathematical band-limited (and thus real-analytic) domain is standard practice in both bioelectrical 
\cite{Sornmo2005} and mechanical \cite{Randall2011} signal processing.

\section{Heuristic Derivation of a Wigner-Ville Approximant for Chirps}
\label{sec:heuri}

Having implemented the theoretical framework described above, we recognize the need to develop analytical evaluation techniques for the PWVD of chirps to gain 
comprehensive insights into the properties of the Wigner surface in the time-frequency domain.
To this end, we derive a heuristic approximation method that enables systematic analysis of the PWVD; this method will be presented in the following sections.

For the sake of simplicity, we resort to the WVD to derive the heuristic approximation. 
The WVD is particularly convenient because it does not rely on a windowing function. 
In fact, introducing a window only complicates the calculation, while ultimately contributing nothing more than a local
multiplicative factor in the expansion. From a conceptual standpoint, this choice
 does not limit the generality of the analysis, as the heuristic derivation already captures the essential mechanisms 
 underlying mathematically more elaborate approaches, even if it is less numerically accurate than the full expansion. In this sense, the WVD 
 serves as a natural starting point, providing both analytical tractability and the key insights necessary for a deeper 
 understanding of mathematically founded derivations.

Let us consider a general nonlinear chirp signal of the form $s(t) = a(t)e^{i\phi(t)}$. The WVD of this signal is defined as:
\begin{equation}
W_s(t, \omega) = \int_{-\infty}^{+\infty} s\left(t + \frac{\tau}{2}\right) s^*\left(t - \frac{\tau}{2}\right) e^{-i\omega\tau} \,d\tau
\end{equation}

\paragraph{Step 1: Phase and Amplitude Approximations.} The core of the heuristic lies in expanding the signal's phase difference $\phi(t + \tau/2) - \phi(t - \tau/2)$ in a Taylor series. Truncating at the cubic term gives $\phi'(t)\tau + \frac{\phi'''(t)}{24}\tau^3$. Assuming $a(t)$ is slowly-varying, the amplitude product is $a(t+\tau/2)a(t-\tau/2) \approx a(t)^2$.

\paragraph{Step 2: Integral and Connection to Airy.} Combining these simplifications, the WVD integral becomes:
\begin{equation}
W_s(t, \omega) \approx a(t)^2 \int_{-\infty}^{+\infty} e^{i\left[ (\phi'(t)-\omega)\tau + \frac{\phi'''(t)}{24}\tau^3 \right]} \,d\tau
\label{eq:heuristic_integral}
\end{equation}
This integral matches the structure of the integral representation of the Airy function, $\text{Ai}(x)$. Following the standard definition cited by Berry via Abramowitz \& Stegun \cite{Abramowitz1964}, a change of variables on Eq. \eqref{eq:heuristic_integral} yields the final, asymptotic formula:
\begin{equation}
W_{s, \text{heuristic}}(t, \omega) = \frac{4\pi a(t)^2}{|\phi'''(t)|^{1/3}} \text{Ai}\left( \frac{2(\phi'(t)-\omega)}{(\phi'''(t))^{1/3}} \right).
\label{eq:heuristic_corrected}
\end{equation}
In summary, the heuristic derivation shows that the local structure of the WVD is governed by the Airy function. 
However, Eq.~(\ref{eq:heuristic_corrected}) is only valid in the vicinity of the instantaneous frequency curve (is a {\it local approximation}), and its accuracy deteriorates away from it due to the slowly-varying amplitude assumption. 
We will now show that this Airy-type behavior is not merely an artifact of this local expansion. Instead, it is the manifestation of a deeper physical principle, 
which we will explore in the next section through the lens of M.V. Berry's semi-classical quantum mechanics geometric framework.

\section{Semi-Classical Interpretation}
\label{sec:semiclassical}

After presenting the heuristic derivation based on a local Taylor expansion of the phase, which leads to an Airy-function-based asymptotic formula, we now turn to a broader and  physically based semi-classical framework valid in the quantum mechanics. 
This section reveals how the local approximation fits naturally within Berry's geometric construction \cite{Berry1977}, providing a unified physical and mathematical interpretation of the PWVD near the instantaneous frequency curve.

The heuristic result of Eq. \eqref{eq:heuristic_corrected} is not an independent approximation but is the precise local limit of the more general framework developed by M.V. Berry in the context of semi-classical 
approximation of quantum regime.
The core of Berry's method is a geometric construction involving a {\it chord}. We replicate this in the time-frequency plane.

\subsection{The Geometric Approximation}
Berry's approach provides a geometric approximation for some integral arising in semi-classical quantum mechanics \cite{Berry1977} similar to the WVD, and valid even far from the instantaneous frequency (IF)  curve.
Translating Berry formula in the time-frequency domain we find:
\begin{equation}
W_s(t, \omega) \approx \mathcal{C}_b(t, \omega) \cdot \text{Ai}\left( -\left[\frac{3}{2}\mathcal{A}_b(t, \omega)\right]^{\frac{2}{3}} \right).
\label{eq:berry_uniform}
\end{equation}
Here, $\mathcal{A}_b(t, \omega)$ is the geometric area enclosed by the {\it Berry chord}, displayed in Fig. \ref{fig:berry}, and the IF curve $\omega_i(t) = \phi'(t)$, and $\mathcal{C}_b(t, \omega)$ is an amplitude 
factor that depends on the signal's power and the local geometry of the IF curve.

\subsection{The Berry Chord}
The {\it classical trajectory} for a signal is its IF curve, given by $\omega = \omega_i(t)$. 
To evaluate the WVD at a point $(t, \omega)$, we construct a chord with the following properties:
\begin{enumerate}
    \item The chord is a straight line segment that intersects the IF curve at two points, $(t_1, \omega_1)$ and $(t_2, \omega_2)$, where $\omega_1 = \omega_i(t_1)$ and $\omega_2 = \omega_i(t_2)$.
    \item The point $(t, \omega)$ is the exact midpoint of this chord.
\end{enumerate}
In Fig. \ref{fig:berry} is displayed a pictorial description of the concepts introduced in this section. 
Mathematically, the midpoint condition is:
\begin{equation}
t = \frac{t_1 + t_2}{2} \quad \text{and} \quad \omega = \frac{\omega_1 + \omega_2}{2} = \frac{\omega_i(t_1) + \omega_i(t_2)}{2}.
\label{eq:midpoint}
\end{equation}
For a given $t$ and $\omega$, a solution to equations (\ref{eq:midpoint}) must be found. 
The \ref{app:cordareale} presents the Theorem \ref{th:cordareale} that proves the local existence of real solutions for the concave or convex case.

\begin{figure}[h!]
    \centering
    \includegraphics[width=1\textwidth]{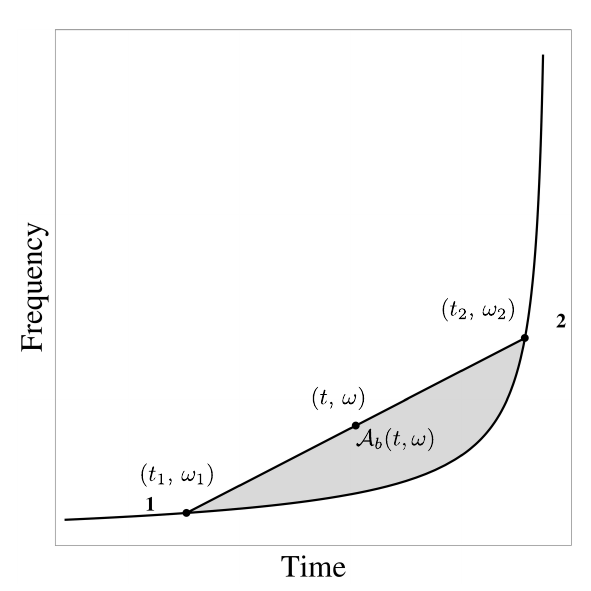}
    \caption{A schematic of the time-frequency trajectory for a frequency-modulated (chirp) signal. The initial state $(t_1, \omega_1)$ and final state $(t_2, \omega_2)$ are highlighted.
     The Berry chord (straight line segment) and the Berry area $\mathcal{A}_b$ (shaded region) are fundamental geometric 
     quantities for the asymptotic analysis of the signal and the characterization of its phase.}
    \label{fig:berry}
\end{figure}

\subsection{The Area $\mathcal{A}_b(t, \omega)$}
The area $\mathcal{A}_b(t, \omega)$ is defined as the geometric area enclosed between the arc of the IF curve from $t_1$ to $t_2$ and the Berry chord itself. Formally,
\begin{equation}
\mathcal{A}_b(t, \omega) = \int_{t_1}^{t_2} \omega_i(\tau) \, d\tau - \omega \cdot (t_2 - t_1).
\label{eq:area}
\end{equation}

It is important to stress that this area should be regarded as a \emph{positive quantity}. 
The oscillatory or exponentially decaying character of the Airy function in Eq.~\eqref{eq:berry_uniform} is not controlled by the sign of $\mathcal{A}_b$, but rather by the location of the evaluation point $(t,\omega)$ with respect to the IF curve. 
When $(t,\omega)$ lies in the classically allowed (convex/concave) region, the corresponding Airy function oscillates. 
By contrast, when $(t,\omega)$ is in the classically forbidden region (the {\it tunneling} regime in quantum mechanics), the Airy function exhibits exponential decay. 

To describe this latter regime within Berry's framework, one must analytically continue the definition of the chord into the complex plane. 
In this case the simple geometric interpretation of $\mathcal{A}_b(t, \omega)$ as a real area is lost, but the uniform approximation still holds \cite{chester1957} by analytic continuation
 (see \ref{app:bleistein}). 
This parallels the well-known semiclassical treatment of quantum tunneling, where forbidden regions are accessed by deforming trajectories into the complex domain.

It is worth noting that, in the language of signal processing, the geometric area $\mathcal{A}_b(t, \omega)$ has a direct interpretation as the phase of the signal accumulated along the chord. 
This correspondence further clarifies the role of $\mathcal{A}_b(t, \omega)$ as the key quantity governing the oscillatory or exponentially decaying behavior of the PWVD.

\subsection{The Heuristic as a  Short-Chord Limit}
The heuristic approach (see section \ref{sec:heuri}) is valid only for points $(t, \omega)$ very close to the IF curve. 
This condition implies that the corresponding Berry chord is short, meaning its duration $\tau_0$ is small. We now show, 
step-by-step, that the geometric formulation \eqref{eq:berry_uniform} reduces to the heuristic formulation \eqref{eq:heuristic_corrected} in this limit.

\paragraph{Step 1: Expansion of the Geometric Area.}
The area $\mathcal{A}_b(t, \omega)$ is exactly equal to the Wigner phase evaluated at the stationary point $\tau_0$ (i.e. the chord duration). 
Expanding for small $\tau_0$ gives:
\begin{equation}
\mathcal{A}_b(t, \omega) = \phi\left(t + \frac{\tau_0}{2}\right) - \phi\left(t - \frac{\tau_0}{2}\right) - \omega\tau_0 \approx (\phi'(t) - \omega)\tau_0 + \frac{\phi'''(t)}{24}\tau_0^3
\label{eq:area_expanded}
\end{equation}

\paragraph{Step 2: Expansion of the Stationary Phase Condition.}
The chord duration $\tau_0$ is determined by the stationary phase condition. Expanding the terms $\phi'(t \pm \tau_0/2)$ in a Taylor series up to the second order gives the key relationship between the frequency deviation and the chord duration:
\begin{equation}
\frac{\phi'\left(t+\frac{\tau_0}{2}\right) + \phi'\left(t-\frac{\tau_0}{2}\right)}{2} = \omega \implies \phi'(t) + \frac{\phi'''(t)}{8}\tau_0^2 \approx \omega
\end{equation}
Rearranging, we find:
\begin{equation}
\omega - \phi'(t) \approx \frac{\phi'''(t)}{8}\tau_0^2
\label{eq:key_relation}
\end{equation}

\paragraph{Step 3: Relating Area to Frequency Deviation.}
We now use Eq. \eqref{eq:key_relation} to eliminate the intermediate variable $\tau_0$ from the area expression. First, substitute $(\phi'(t)-\omega) = -\frac{\phi'''(t)}{8}\tau_0^2$ into Eq. \eqref{eq:area_expanded}:
\begin{align}
\mathcal{A}_b(t, \omega) &\approx -\frac{\phi'''(t)}{12}\tau_0^3
\label{eq:area_final_approx}
\end{align}
Next, we isolate $\tau_0^3$ from Eq. \eqref{eq:key_relation} by raising it to the power of $3/2$:
\[ \tau_0^3 = (\tau_0^2)^{3/2} \approx \left(\frac{8(\omega - \phi'(t))}{\phi'''(t)}\right)^{3/2} = \frac{16\sqrt{2} (\omega - \phi'(t))^{3/2}}{(\phi'''(t))^{3/2}} \]
Substituting this expression for $\tau_0^3$ into Eq. \eqref{eq:area_final_approx}, we get:
\begin{equation}
\mathcal{A}_b(t, \omega) \approx -\frac{4\sqrt{2}}{3} \frac{(\omega - \phi'(t))^{3/2}}{(\phi'''(t))^{1/2}}.
\end{equation}

\paragraph{Step 4: Final Comparison of the Arguments.}
We now substitute this approximate area into the argument of the geometric formulation \eqref{eq:berry_uniform}:
\begin{align}
-\left[\frac{3}{2}\mathcal{A}_b(t, \omega)\right]^{\frac{2}{3}} &\approx -\left[\frac{3}{2} \left(-\frac{4\sqrt{2}}{3} \frac{(\omega - \phi'(t))^{3/2}}{(\phi'''(t))^{1/2}}\right)\right]^{\frac{2}{3}} \nonumber \\
&= -\left[-2\sqrt{2} \cdot \frac{(\omega - \phi'(t))^{3/2}}{(\phi'''(t))^{1/2}}\right]^{\frac{2}{3}}
\end{align}
Assuming $\omega > \phi'(t)$ and $\phi'''(t) > 0$ for simplicity, the inner term is negative. The overall real-valued result is:

\beq
-\left[\frac{3}{2}\mathcal{A}_b(t, \omega)\right]^{\frac{2}{3}} \sim - \frac{2(\omega - \phi'(t))}{(\phi'''(t))^{1/3}}.
\eeq
This is {\it exactly identical} to the argument of the Airy function in our heuristic derivation Eq. (\ref{eq:heuristic_corrected}).
In the derivation above we assumed $\phi'''(t) > 0$ for simplicity. If instead $\phi'''(t) < 0$, the analysis follows identically: the cubic root $(\phi'''(t))^{1/3}$ remains well defined (once a consistent branch is chosen), leaving the final Airy argument unaffected.
The only effect of changing the sign of $\phi'''(t)$ is to reverse which side of the IF curve corresponds to the oscillatory regime and which to the exponentially decaying regime.

\paragraph{Step 5: the Amplitude Pre-factors.}
We now reformulate Berry's original formula, equation 4.7 from his paper \cite{Berry1977}:
\begin{equation}
\Psi(q,p) = \frac{2}{\sqrt{\pi\hbar |D(q,p)|}} \text{Ai}\left( -\left[\frac{3\mathcal{A}_b(q,p)}{2\hbar}\right]^{\frac{2}{3}} \right)
\label{eq:berry_original}
\end{equation}
where $D(q,p)$ is a denominator related to the action derivatives.
The argument of the Airy function translates directly as shown before (let us set $\hbar = 1$). 
Translating the denominator $|D(q,p)|$ is non-trivial. However, Berry's own transitional approximation (eq. 4.8 Ref. \cite{Berry1977}) shows that near the trajectory, 
the amplitude is dominated by its local curvature. For the IF curve $\omega_i(t)$, the curvature is related to its second derivative, $\omega_i''(t)$, which is equivalent to the third derivative of the signal's phase, $\phi'''(t)$.

The amplitude factor $\mathcal{C}_b(t, \omega)$ can therefore be approximated as:
\begin{equation}
\mathcal{C}_b(t, \omega) \propto \frac{1}{\left|\phi'''(t)\right|^{1/3}}.
\label{eq:amplitude}
\end{equation}
We note that in the formula (\ref{eq:amplitude}) the chirp amplitude is not present.
Indeed,  Berry's original semiclassical formula does not explicitly contain a term equivalent to the amplitude of a chirp signal.

In conclusion, the connection between the PWVD and Berry's semiclassical chord construction is exact at the phase level, but only partial at the amplitude level, 
and the geometric framework must be extended via analytic continuation to be globally valid. 
Indeed, as discussed above, accessing the evanescent (tunneling) regime below the instantaneous frequency curve requires evaluating the chord in the complex plane. 

Furthermore, Berry's foundational integral, designed for infinite-domain quantum probability distributions, lacks a native mechanism to handle the amplitude modulation and finite windowing inherent to practical AM-FM signals. This highlights the necessity of our approach: while the geometric analogy beautifully captures the phase behavior, the correct amplitude pre-factors and boundary truncation effects must be explicitly derived from the signal model. 

Having established the connection between the local heuristic and Berry's broader geometric framework, we are now ready to proceed to the main result of this paper. The following section provides the mathematical derivation of the uniform asymptotic formula, which formalizes these physical and geometric insights into a powerful analytical tool.

\section{Uniform Asymptotic Evaluation Techniques}
\label{sec:uniform}

This section presents a complete derivation of a uniform asymptotic expansion \cite{chester1957, borovikov1994} (see \ref{app:bleistein}) 
for the PWVD of a chirp signal.
The analysis is based on a crucial assumption, common in signal analysis practice: the use of a windowing function that vanishes smoothly at the edges of its support.

To formalize the problem, let us consider an analytic signal $x(t) = a(t)e^{\mathrm{i} \chi\psi(t)}$ and its PWVD:
\beq
PW_x(t,f)=\int_{-T}^{T} w\left(\frac{\tau}{2}\right) w\left(-\frac{\tau}{2}\right) x\left(t+\frac{\tau}{2}\right)x^{\ast}\left(t-\frac{\tau}{2}\right)e^{-\mathrm{i} 2\pi f \tau} d\tau.
\eeq
The integral can be written in the canonical form :
\beq
PW_x(t,\bar{f})=\int_{-T}^{T} A(\tau,t) e^{\mathrm{i}\chi F(\tau,t,\bar{f}) }d\tau.
\label{eq:canonical}
\eeq
where $\chi$ is a large parameter, and we define:
\beq
A(\tau, t) = w(\frac{\tau}{2}) w(-\frac{\tau}{2}) a(t + \frac{\tau}{2}) a(t - \frac{\tau}{2})
\eeq
\beq
F(\tau, t, \bar{f}) = \psi(t + \frac{\tau}{2}) - \psi(t - \frac{\tau}{2}) - 2\pi \bar{f} \tau.
\label{eq:faseF}
\eeq
where $\bar{f}= f/{\chi}$ is the asymptotic frequency.

Our {\it key assumption} is that the window $w(\cdot)$ is a \emph{tapered} function, such that it vanishes at the edges of the integration interval:
\beq
w\left(\frac{T}{2}\right) = w\left(-\frac{T}{2}\right) = 0.
\eeq

We emphasize that this assumption does not exclude commonly used window functions; rather, it is perfectly aligned with standard signal processing practices. 
Widespread (continuous-time) tapered windows, such as Hann, Tukey (when the decay parameter is strictly greater than zero), or Blackman windows~\cite{Harris1978, PrabhuBook2014}, are specifically designed to smoothly decay to zero at their boundaries to minimize spectral leakage, thereby naturally satisfying this boundary condition. As we will demonstrate, this physical tapering is what simplifies the mathematical structure of both the asymptotic expansion and the remainder.

The functions $\Phi_{inc}$ and $\Phi'_{inc,p}$ are the foundations of our expansion. They are the generalizations of the standard Airy function $\text{Ai}(p)$ and its derivative $\text{Ai}'(p)$ for a finite and symmetric integration interval.

We define the Symmetric Incomplete Airy Function as the Fourier integral representation of the Airy function (see details in the paper \cite{cwik1988efficient, levey1969incomplete}), but evaluated over a finite interval $[-s, s]$:
\beq
\Phi_{inc}(p,s) := \frac{1}{2\pi} \int_{-s}^{s} e^{i\left(\frac{t^3}{3} + pt\right)} dt.
\eeq
This definition has two fundamental properties:
\begin{itemize}
    \item It is a real function, since the integrand has an even real part and an odd imaginary part which cancels out over a symmetric integration range.
    \item In the limit where the integration interval becomes infinite ($s \to \infty$), it converges to the standard Airy function: $\lim_{s \to \infty} \Phi_{inc}(p,s) = \text{Ai}(p)$.
\end{itemize}

\subsection{Uniform Expansion Theorem for PWVD}
In this section we develop a uniform  expansion of PWVD integral using  uniform approximation for oscillatory integrals with two coalescing stationary phase points (see Ref.s \cite{chester1957, borovikov1994} and \ref{app:bleistein}). 
To the best of our knowledge, uniform asymptotic expansion techniques, which are widely employed in electromagnetics for the evaluation of diffraction integrals \cite{felsen1973radiation},
 have not previously been applied within the context of signal processing. This highlights the novelty of the present approach.

To establish the validity of the uniform asymptotic expansion for the PWVD, we must verify that its integral representation, 
as formulated in Eq.~\eqref{eq:canonical}, satisfies the key hypotheses of the theorem by Chester, Friedman, and Ursell, detailed in \ref{app:bleistein}. 
The theorem is tailored for oscillatory integrals where the phase function possesses two stationary points that merge into a single, degenerate stationary point.

\paragraph{Mapping of Variables.} First, we map the variables from the general theorem to our specific PWVD integral.
\begin{itemize}
    \item The integration variable, denoted by $x$ in the theorem, is $\tau$.
    \item The large asymptotic parameter denoted by $\lambda$ is replaced by $\chi$.
    \item The phase function denoted by $\varphi(x, a)$, corresponds to $F(\tau, t, \bar{f})$ from Eq.~\eqref{eq:faseF}.
    \item The critical \textit{coalescence parameter}, denoted by $a$ in the theorem, is the analysis frequency $f$. 
    Its variation controls the merging of the stationary points $\pm\tau_s$ at $\tau=0$ as $f$ approaches the instantaneous frequency $f_i(t)$.
    \item The time $t$ acts as a \textit{spectator parameter}, which is held constant during the asymptotic analysis at a specific point $(t,f)$ but defines the local properties of the signal.
\end{itemize}

The theorem's applicability rests on two non-degeneracy conditions, which we verify at the coalescence point ($\tau=0$, $f = f_i(t)$).

\paragraph{First Condition: Non-vanishing Third Derivative.}
The theorem requires that the third derivative of the phase with respect to the integration variable be non-zero at the coalescence point.
\[
    \frac{\partial^3 F}{\partial \tau^3} \bigg|_{\tau=0} \neq 0.
\]
We compute this derivative from the phase function $F(\tau, t, \bar{f}) = \psi(t + \tau/2) - \psi(t - \tau/2) - \frac{2\pi f \tau}{\chi}$. The third derivative with respect to $\tau$ is:
\[
    \frac{\partial^3 F}{\partial \tau^3} = \frac{1}{8}\psi'''\left(t+\frac{\tau}{2}\right) - \left(-\frac{1}{8}\right)\psi'''\left(t-\frac{\tau}{2}\right) = \frac{1}{8}\left[ \psi'''\left(t+\frac{\tau}{2}\right) + \psi'''\left(t-\frac{\tau}{2}\right) \right].
\]
Evaluating at $\tau=0$, we obtain:
\begin{equation}
    \frac{\partial^3 F}{\partial \tau^3} \bigg|_{\tau=0} = \frac{1}{4}\psi'''(t).
\end{equation}
Thus, the first condition is satisfied \textit{provided} that $\psi'''(t) \neq 0$ i.e.  $f_i''(t) \neq 0$. 
This is a mild and physically meaningful constraint, requiring that the instantaneous frequency curve $f_i(t)$ does not have an inflection point at the time of analysis $t$.

\paragraph{Second Condition: Non-vanishing Mixed Derivative.}
The theorem's second crucial hypothesis is that the mixed partial derivative with respect to the integration variable and the coalescence parameter be non-zero:
\[
    \frac{\partial^2 F}{\partial \tau \partial f} \bigg|_{\tau=0} \neq 0.
\]
We first compute the derivative of the phase with respect to $\tau$:
\[
    \frac{\partial F}{\partial \tau} = \frac{1}{2}\psi'\left(t+\frac{\tau}{2}\right) - \left(-\frac{1}{2}\right)\psi'\left(t-\frac{\tau}{2}\right) - \frac{2\pi f}{\chi} = \frac{1}{2}\left[\psi'\left(t+\frac{\tau}{2}\right) + \psi'\left(t-\frac{\tau}{2}\right)\right] - \frac{2\pi f}{\chi}.
\]
Next, we differentiate this expression with respect to the parameter $f$:
\begin{equation}
    \frac{\partial^2 F}{\partial \tau \partial f} = \frac{\partial}{\partial f} \left( \frac{1}{2}\left[\psi'\left(t+\frac{\tau}{2}\right) + \psi'\left(t-\frac{\tau}{2}\right)\right] - \frac{2\pi f}{\chi} \right) = -\frac{2\pi}{\chi}.
\end{equation}
Since the asymptotic parameter $\chi$ is non-zero, this mixed partial derivative is a non-zero constant. Therefore, this second condition is \textit{unconditionally satisfied} due to the fundamental structure of the Wigner-Ville integral itself.

Both foundational hypotheses of the theorem are met. The first imposes a mild physical constraint on the signal's phase ($f_i''(t) \neq 0$ i.e. nonlinear chirps), while the second is structurally guaranteed. 
This provides a complete justification for the application of Bleistein's uniform asymptotic method to the PWVD of nonlinear chirp signals.


Let us focus on the asymptotic expansion of the PWVD integral. Its most general form for an integral over a finite interval 
$[a, b]$ includes contributions from the internal stationary points (Airy functions) and explicit additive boundary contributions from the endpoints $a$ and $b$
 (often denoted by $\mathcal{C}$, see \ref{app:bleistein}).

However, when the integrand's amplitude vanishes at the endpoints (thanks to the assumption on the window $w(t)$), 
the explicit boundary term $\mathcal{C}$ possesses a peculiar property, illustrated below. 
While it does not vanish completely, its contribution is strictly bounded by $\order{\chi^{-1}}$ 
and it vanishes as the  half‑width of the spectral window goes to infinity. (as proven in  \ref{app:remainder_bound}). 

Consequently, rather than keeping a cumbersome explicit boundary term, we absorb it into the global remainder $R(\chi)$. 
The asymptotic expansion thus simplifies to an elegant form, formally analogous to that over an infinite interval:
\begin{multline}
PW_x(t, \bar{f}) = \left(\frac{2\pi}{\chi^{1/3}}\right) e^{i\chi \phi_0} \left\{ A_{series} \Phi_{inc}(-\chi^{2/3}\xi, \chi^{1/3}q) \right. \\
- \left. \frac{iB_{series}}{\chi^{1/3}} \text{sgn}(\frac{\partial^3 F}{\partial \tau^3 } \bigg|_{\tau=0} ) \Phi'_{inc,p}(-\chi^{2/3}\xi, \chi^{1/3}q) \right\} + R(\chi)
\label{eq:bleistein_simplified}
\end{multline}
where $A_{series}$ and $B_{series}$ are complete asymptotic series, the base functions are the incomplete Airy functions ($\Phi_{inc}$) and their derivatives,
 and the global remainder $R(\chi)$.

In Eq. (\ref{eq:bleistein_simplified}) $q$ is the transformed integration limit, defined by $F(T,t,\bar{f}) = q^3/3 - \xi q$.


The sign into equation (\ref{eq:bleistein_simplified}) can be evaluated taking into account eq. (\ref{eq:faseF}) and computing the derivative in particular for $\tau=0$, we get:
\[
\frac{\partial^3 F}{\partial \tau^3} \bigg|_{\tau=0}  = \tfrac{1}{4}\,\psi'''(t),
\]
that can be written  in terms of IF
\[
\frac{\partial^3 F}{\partial \tau^3} \bigg|_{\tau=0}  = \frac{\pi}{2\chi}\, f_i''(t).
\]
We now proceed to calculate the dominant terms of the expansion \eqref{eq:bleistein_simplified}.

\subsection{Stationary Phase Points and Phase Parameters}

The critical points $\pm\tau_s$ are the solutions of the \emph{Berry's chord}
\beq
\frac{f_i\left(t + \frac{\tau_s}{2}\right) + f_i\left(t - \frac{\tau_s}{2}\right)}{2} = f
\eeq
where $f_i(t) = \frac{\chi}{2\pi} \psi'(t)$ is the instantaneous frequency of the signal.
The phase parameters $\phi_0$ and $\xi$ are determined by the geometry of the phase $F(\tau)$:
\begin{itemize}
    \item The midpoint phase $\phi_0$ is zero due to symmetry.
    \item The parameter $\xi$ is related to the phase difference: $\xi = \left[ \frac{3}{2} F(\tau_s, t, \bar{f}) \right]^{2/3}$.
\end{itemize}

It is crucial to highlight an implicit but fundamental assumption underpinning this entire framework: the instantaneous frequency function $f_i(t)$ must be a \textit{real analytic function}. 
This condition is not merely a convenience but a core requirement with a dual justification, rooted in both the geometric interpretation and the underlying mathematical machinery.

From a geometric perspective, analyticity allows the concept of the Berry's chord to be extended from the real domain to the complex plane via \textit{analytic continuation}. 
This unifies the oscillatory 
regime (where the chord's endpoints $\pm\tau_s$ are real) and the evanescent or {\it tunneling} regime (where the endpoints become a complex-conjugate pair), ensuring that Eq.~\eqref{eq:midpoint} 
remains meaningful across a broad region of the time–frequency plane.

From a methodological standpoint, the justification is even more stringent. The uniform asymptotic expansion itself, derived using the method of \textit{canonical cubic transformation} (see \cite{chester1957, levinson1960canonical} and \ref{app:bleistein}), relies on a critical change of variables. 
This transformation maps the original phase function $F(\tau, t, \bar{f})$ to a canonical cubic polynomial. The existence of such a mapping as an 
\textit{analytic diffeomorphism}---which is essential for the resulting asymptotic series to be well-defined---is guaranteed only if the original phase function $F(\tau, t, \bar{f})$
 is itself analytic in the integration variable $\tau$. Since the $\tau$-dependence in $F$ arises from the terms $\psi(t \pm \tau/2)$, this requirement directly translates to the condition that the signal's phase $\psi(t)$, 
 and thus its derivative $f_i(t)$, must be real analytic.

Therefore, the analyticity of $f_i(t)$ is a foundational prerequisite for the validity of the uniform asymptotic approximation. As established in subection~\ref{sec:phase_analyticity}, 
this condition is satisfied by broad and physically relevant classes of signals.

\subsection{The Leading Coefficients $A_0$ and $B_0$}
Thanks to the symmetry of the PWVD problem (odd phase and even amplitude in $\tau$), the leading coefficients of the series $A_{series}$ and $B_{series}$ are:
\begin{itemize}
    \item $B_0 = 0$: Symmetry cancels the leading term of the $B_{series}$.
    \item $A_0 \neq 0$: This is the dominant term of the entire expansion. Its expression is:
    \beq
    A_0 = \xi^{1/4} A(\tau_s, t) \left| \frac{2}{\frac{\partial^2 F}{\partial \tau^2}(\tau_s, t, \bar{f})} \right|^{1/2}.
    \eeq
\end{itemize}
To make this formula operational, we write out the terms $A(\tau_s, t)$ and $F''(\tau_s, t, \bar{f})$ as functions of the signal's properties:
\begin{align}
A(\tau_s, t) &= w^2\left(\frac{\tau_s}{2}\right) a\left(t+\frac{\tau_s}{2}\right) a\left(t-\frac{\tau_s}{2}\right) \\
\frac{\partial^2 F}{\partial \tau^2}(\tau_s, t, \bar{f}) &= \frac{\pi}{2\chi} \left[ f_i'\left(t+\frac{\tau_s}{2}\right) - f_i'\left(t-\frac{\tau_s}{2}\right) \right]
\end{align}

\subsection{Final Uniform Expression}
By substituting the calculated coefficients into the simplified expansion \eqref{eq:bleistein_simplified} and keeping only the dominant term (with $A_0$), we obtain the final approximation for the PWVD:

\begin{multline}
PW_x(t, f) \approx \left(\frac{2\pi}{\chi^{1/3}}\right) \Phi_{inc}(-\chi^{2/3}\xi, \chi^{1/3}q) \times \\
\left\{ \xi^{\frac{1}{4}} \cdot w^2\left(\frac{\tau_s}{2}\right) a\left(t+\frac{\tau_s}{2}\right) 
a\left(t-\frac{\tau_s}{2}\right) \cdot \left| \frac{4\chi}{\pi \left( f_i'(t+\frac{\tau_s}{2}) - f_i'(t-\frac{\tau_s}{2}) \right)} \right|^{\frac{1}{2}} \right\}
\label{eq:BUAPWV}
\end{multline}
where $q$ is the transformed integration limit, defined by $F(T,t,\bar{f}) = q^3/3 - \xi q$.


\subsection{Comment on the Asymptotic Accuracy of the Result}

It is important to emphasize the nature of this approximation and its error bounds. The formula (\ref{eq:BUAPWV}) represents the \textit{leading term} of the complete formal asymptotic expansion (\ref{eq:bleistein_simplified}). 

Thanks to the remainder estimate derived in \ref{app:remainder_bound}, we can explicitly differentiate the total asymptotic error into two distinct mathematical categories, governed by different physical parameters:

\begin{enumerate}
    \item \textbf{Intrinsic Asymptotic Series Error:} In the classical theory of uniform asymptotic expansions \cite{chester1957}, the error is governed by the first neglected term of the series ($A_1$). This term introduces a correction of order $\mathcal{O}(\chi^{-4/3})$. This error is intrinsic to the local cubic approximation of the phase and the Taylor expansion of the amplitude near the caustics. Crucially, this error depends strictly on the asymptotic parameter $\chi$ and \textit{persists even if the window size becomes infinite} ($T \to \infty$).
    
    \item \textbf{Boundary Truncation Error:} The PWVD is inherently a finite-domain integral. As demonstrated in \ref{app:remainder_bound}, the truncation at boundaries $\pm T$ introduces a global remainder $R(\chi)$ bounded by $\mathcal{O}(\chi^{-1})$. However, a deeper look at the explicit boundary evaluation of the 
residual function (see discussion in \ref{app:remainder_bound})  reveals a profound physical property. 
 The parameter $q$ is the transformed integration limit, which grows monotonically with the window size $T$. 
 Therefore, the magnitude of this boundary error is inversely proportional to the window support. 
\end{enumerate}

This dichotomy beautifully explains the transition from the PWVD to the standard WVD. For a fixed, finite window $w(\tau)$, the global error is technically dominated by the boundary term $\mathcal{O}(\chi^{-1})$. However, as the window size increases ($T \to \infty$), the boundary truncation error asymptotically vanishes to zero. 
In this limit, the $\mathcal{O}(\chi^{-1})$ error disappears, and the accuracy of the expansion is entirely governed by the intrinsic $\mathcal{O}(\chi^{-4/3})$ Airy series error, 
perfectly recovering the classical infinite-domain Wigner-Ville theory.

For practical applications where both the asymptotic parameter $\chi$ and the window size $T$ are sufficiently large, the leading-order approximation derived here offers excellent accuracy, effectively capturing the fundamental structure of the PWVD without being spoiled by boundary artifacts.


\subsection{Partial Conclusions}

We have derived a uniform asymptotic expression for the PWVD under the practical assumption of using a windowing function that vanishes at the edges of its support. 
This condition ensures that the explicit boundary contributions are bounded by $\order{\chi^{-1}}$ and can be absorbed into the remainder. 
This leads to a clean final formula where the PWVD is described solely by incomplete Airy functions modulated by coefficients that depend directly on the physical properties of the signal. 
This establishes a clear hierarchy among the approximations: our general uniform formula Eq.~(\ref{eq:BUAPWV}) serves as the most accurate description for practical, finite-windowed signals. 
It converges to the heuristic formula Eq.~(\ref{eq:heuristic_corrected}) under the combined limits of a wide window---where $\Phi_{\text{inc}}$ converges to $\text{Ai}$---and a local analysis in the neighborhood of the IF curve, 
where the coefficients simplify to their Taylor-approximated form. 
Finally, this result not only provides a powerful analytical tool but also clarifies the fundamental role of tapered windows in making the 
practical analysis over finite intervals converge to the elegant form of the mathematical theory over infinite intervals.

\section{Systematic Bias in Frequency Estimation}
 \label{sec:bias}

A critical issue in peak detection using the PWVD is the bias affecting the estimation of the IF.
This bias is not incidental but inherent to the estimation process, and its nature can be rigorously investigated through the
 uniform asymptotic expansion derived in Section~\ref{sec:uniform},  via our final expression Eq.~(\ref{eq:BUAPWV}).

To quantify this bias, we must find the  frequency, $f_{\text{max}}(t)$,
that maximizes the uniform asymptotic expression for the PWVD. The bias is then defined as $\text{Bias}(t) = f_{\text{max}}(t) - f_i(t)$.

The expression to maximize is:
\begin{equation}
    PW_x(t, f) \approx C(t, f) \cdot \Phi_{\text{inc}}(p(f), s(f))
\end{equation}
where $p(f) = -\chi^{2/3}\xi(f)$ and $s(f) = \chi^{1/3}q(f)$. To find the maximum, we set the derivative with respect to $f$ to zero. Applying the product and chain rules yields:
\begin{equation}
    \frac{d}{df} PW_x = \frac{dC}{df}\Phi_{\text{inc}} + C \left( \frac{\partial\Phi_{\text{inc}}}{\partial p}\frac{dp}{df} + \frac{\partial\Phi_{\text{inc}}}{\partial s}\frac{ds}{df} \right) = 0
\end{equation}
We analyze the terms in the large $\chi$ limit. The amplitude $C(t,f)$ and its derivative are slowly varying and do not scale with positive powers of $\chi$. However, the derivatives of the arguments $p$ and $s$ do:
\begin{itemize}
    \item $\frac{dp}{df} = -\chi^{2/3} \frac{d\xi}{df}$, which scales as $\mathcal{O}(\chi^{2/3})$.
    \item $\frac{ds}{df} = \chi^{1/3} \frac{dq}{df}$, which scales as $\mathcal{O}(\chi^{1/3})$.
\end{itemize}
For large $\chi$, the term containing $\chi^{2/3}$ is asymptotically dominant. Therefore, the condition for the maximum simplifies to the leading order equation:
\begin{equation}
    C \cdot \frac{\partial\Phi_{\text{inc}}}{\partial p}\left(-\chi^{2/3} \frac{d\xi}{df}\right) = 0 \quad \implies \quad \frac{\partial\Phi_{\text{inc}}}{\partial p} \Big|_{(p_{\text{max}}, s)} = 0
\end{equation}
This shows that the peak of the PWVD corresponds to the peak of the incomplete Airy function with respect to its first argument. Let $p_{\text{max}}(s)$ be the value of $p$ that solves this equation. 
The bias is found by relating this value back to the frequency deviation $\xi$. The implicit equation for the peak frequency $f_{\text{max}}$ is therefore:
\begin{equation}
    -\chi^{2/3}\xi(t, f_{\text{max}}) = p_{\text{max}}(s(f_{\text{max}})).
    \label{eq:general_bias}
\end{equation}
This equation is {\it implicit} because the quantity of interest, $f_{\text{max}}$, is embedded within the complex functions $\xi$ and $s$. Due to the highly nonlinear relationship between the frequency $f$ and the coalescence parameter $\xi$, it is not possible to algebraically invert this equation to obtain a simple, explicit formula for $f_{\text{max}}$.

However, we can derive an approximate \textit{explicit asymptotic formula} for the bias. 
The entire framework is based on a large asymptotic parameter $\chi$, which implies that the resulting bias will be small. 
This justifies analyzing the relationship between $\xi$ and $f$ in the local vicinity of the true IF, $f_i(t)$. 
The standard method for this is a first-order Taylor expansion of $\xi(t,f)$ around $f_i(t)$:
\begin{equation}
    \xi(t, f) \approx \xi(t, f_i) + \left. \frac{\partial\xi}{\partial f} \right|_{f=f_i} \cdot (f - f_i(t))
\end{equation}
By definition, $\xi(t, f_i) = 0$. Evaluating this at $f=f_{\text{max}}$, we obtain a local linear relationship between $\xi$ and the bias:
\begin{equation}
    \xi(t, f_{\text{max}}) \approx \left. \frac{\partial\xi}{\partial f} \right|_{f=f_i} \cdot \text{Bias}(t).
\end{equation}
Substituting this linearized relation back into the exact implicit equation (\ref{eq:general_bias}) allows us to solve for the bias:
\begin{equation}
    -\chi^{2/3} \left( \left. \frac{\partial\xi}{\partial f} \right|_{f=f_i} \cdot \text{Bias}(t) \right) \approx p_{\text{max}}(s(f_{\text{max}}))
\end{equation}
This yields the explicit asymptotic formula for the bias:
\begin{equation}
    \text{Bias}(t) \approx - \frac{p_{\text{max}}(\chi^{1/3} q(f_{\text{max}}))}{\chi^{2/3} \left( \left. \frac{\partial\xi}{\partial f} \right|_{f=f_i} \right)}.
    \label{eq:bias_explicit}
\end{equation}
This formula, while being an approximation, is fully consistent with the asymptotic nature of the entire analysis. It reveals that the bias is directly proportional to the location of the Airy peak ($p_{\text{max}}$) and inversely proportional to $\chi^{2/3}$ and the local curvature of the IF curve (which is captured by the derivative $\partial\xi/\partial f$).

\subsection{Calculation of $\partial\xi/\partial f$}
In the short chord limit (where $f \approx f_i$), the coalescence parameter $\xi$ is linearly related to the frequency deviation $f - f_i(t)$. By expanding the defining equations for the stationary point $\tau_s$ and the phase $F(\tau_s,t, \bar{f})$, we find:
\begin{equation}
    \xi(t,f) \approx \left[ \frac{32\pi^2}{\chi^2 f_i''(t)} \right]^{1/3} \cdot (f - f_i(t)).
\end{equation}
The derivative evaluated at the IF is therefore a constant determined by the signal's local properties:
\begin{equation}
    \left. \frac{\partial\xi}{\partial f} \right|_{f=f_i} = \left[ \frac{32\pi^2}{\chi^2 f_i''(t)} \right]^{1/3}.
    \label{eq:xi_derivative}
\end{equation}
Note that $f_i''(t)$ is the second derivative of the instantaneous frequency, i.e., the curvature of the IF curve. In terms of the signal's phase $\psi(t)$, this is related by $f_i''(t) = (\chi/2\pi)\psi'''(t)$.

\subsection{Approximate Explicit Bias Equation}
Substituting the derivative (\ref{eq:xi_derivative}) into the linearized version of the implicit equation (\ref{eq:bias_explicit}) yields the explicit asymptotic formula for the bias:
\begin{equation}
    \text{Bias}(t) \approx - \frac{p_{\text{max}}(\chi^{1/3} q(f_{\text{max}}))}{\chi^{2/3} \left[ \frac{32\pi^2}{\chi^2 f_i''(t)} \right]^{1/3}} = - p_{\text{max}}(\chi^{1/3} q(f_{\text{max}})) \left[ \frac{f_i''(t)}{32\pi^2} \right]^{1/3}.
    \label{eq:general_bias_explicit}
\end{equation}
This formula, consistent with the asymptotic nature of the analysis, reveals that the bias is directly proportional to the location of the incomplete Airy function peak ($p_{\text{max}}$) 
and to the cubic root of the IF curve's curvature ($f_i''(t)$). In the short-chord limit, $p_{\text{max}}$ approaches the constant  which is related  to the peak of the Airy function.

\section{Application to Gravitational Wave Chirps}
\label{sec:GWchirp}

To validate our theoretical framework and demonstrate its practical utility, we apply our uniform asymptotic formula to a gravitational wave (GW) signal from a coalescing binary system. 

Rather than an exhaustive numerical validation across all possible signal classes, this specific highly nonlinear chirp,
of significant scientific interest following its breakthrough discovery \cite{LIGO2016}, serves as an illustrative proof-of-concept. 
It provides an ideal physical laboratory to test our asymptotic approximation in a controlled, high-fidelity regime, where standard time-frequency methods often reach their limits.

In this section, we will first develop an asymptotic model for the whitened GW signal and then use numerical simulations to compare the predictions of our formula against direct computations, 
providing a comprehensive validation of our approach.

The model of a GW signal of a coalescing binary system, based on the lowest-order post-Newtonian (PN) expansion \cite{Sathya},  is a nonlinear chirp described by $s(t) = a(t) \cos(\phi(t))$. Its instantaneous frequency is:
\begin{equation}
f(t) = f_0 \left(1 - \frac{t}{T_c}\right)^{-3/8},
\label{eq:gwfreq}
\end{equation}
and its amplitude is:
\begin{equation}
a(t) = A \Pi(t) f_0^{2/3} \left(1 - \frac{t}{T_c}\right)^{-1/4},
\end{equation}
where $A$ is the amplitude, and $\Pi(t)$ is a rectangular window taking unit value in the range $[0, T_c']$.
We note that $T_c' < T_c$ because the model became inaccurate in the last stage of coalescence, the instantaneous frequency divergence is an artifact of PN expansion.

The crucial element for our analysis is the phase, which contains a large parameter. We define the dimensionless {\it chirp number}, $\chi$, as:
\begin{equation}
\chi = f_0 T_c.
\end{equation}
For typical GW sources, $\chi \gg 1$, ranging from $30$ to $10^3$. This is the large asymptotic parameter that justifies our approach. The signal's phase can be written explicitly in terms of $\chi$:
\begin{equation} 
\phi(t) = \frac{16\pi}{5} \chi \left[1 - \left(1 - \frac{t}{T_c}\right)^{5/8}\right].
\label{eq:GWphase}
\end{equation}

To accurately model the signal's instantaneous properties and effectively apply the Wigner-Ville transform, we construct the analytic representation, $s_a(t)$ of the GW signal $s(t)$. 
Using the asymptotic expansion theorem \ref{th:asinto_appendix} for analytic signals, which is effectively an expansion in powers of $1/\chi$, we express the leading-order analytic signal as:
\begin{equation}
s_a(t) \approx a(t) \, e^{\mathrm{i} \phi(t)}.
\end{equation}

The next step is to filter the signal to whiten the detector's noise. 
The desired output is the whitened analytic signal, $s_w(t)$, which is the convolution of the analytic signal $s_a(t)$ with the whitening filter's impulse response, $h_w(t)$, i.e.
\begin{equation}
s_w(t) = (s_a * h_w)(t).
\end{equation}
In order to computing this convolution directly, we apply the \textit{Adiabatic Filtering Theorem} \ref{th:adiabfilt}.
The whitening filter for a detector like LIGO-I \cite{Sathya} is defined by its frequency response $H_w(f)$:
\begin{equation}
H_w(f) = \frac{1}{\sqrt{\frac{1}{5}\left(\frac{\nu_0}{f}\right)^4 + 2 + 2\left(\frac{f}{\nu_0}\right)^2}},
\end{equation}
where $\nu_0$ is the frequency of the detector's minimum noise.

Applying the theorem's result directly, the whitened analytic signal $s_w(t)$ is given by the input signal's complex amplitude, multiplied by the whitening filter evaluated at the instantaneous frequency:
\begin{equation}
s_w(t) \approx a(t) \, H_w(f(t)) \, e^{\mathrm{i}\phi(t)}.
\end{equation}

By substituting the expressions for $a(t)$, and $H_w(f(t))$,  we arrive at the final, formally-derived model for the whitened analytic gravitational wave signal:
\begin{equation}
s_w(t) \approx \frac{A \Pi(t) f_0^{2/3} \left(1 - \frac{t}{T_c}\right)^{-1/4}}{\sqrt{\frac{1}{5}\left(\frac{\nu_0}{f(t)}\right)^4 + 2 + 2\left(\frac{f(t)}{\nu_0}\right)^2}} \exp\left(\mathrm{i} \phi(t)\right).
\end{equation}
This model provides the leading-order asymptotic representation of the gravitational wave signal (with respect to the large parameter $\chi$) 
after whitening, with its validity rigorously supported by the adiabatic filtering theorem (see  \ref{app:adiabfilt}).


Building on the above theoretical foundations, we next assess its accuracy through numerical simulations. 
Specifically, we analyze the time--frequency properties of the whitened GW chirp using the Pseudo-Wigner--Ville  distribution. 
Figure~\ref{fig:pwvsurvace} displays the numerical PWVD representation computed via FFT with a Hanning window. The relevant parameters are reported in the caption.  
The ridge of the distribution closely follows the theoretical instantaneous frequency, illustrating the ability of the Pseudo-Wigner–Ville representation to track the chirp’s instantaneous frequency.
 This provides a general validation of the method’s effectiveness for analyzing gravitational wave signals.
\begin{figure}[h!]
    \centering
    \includegraphics[width=1\textwidth]{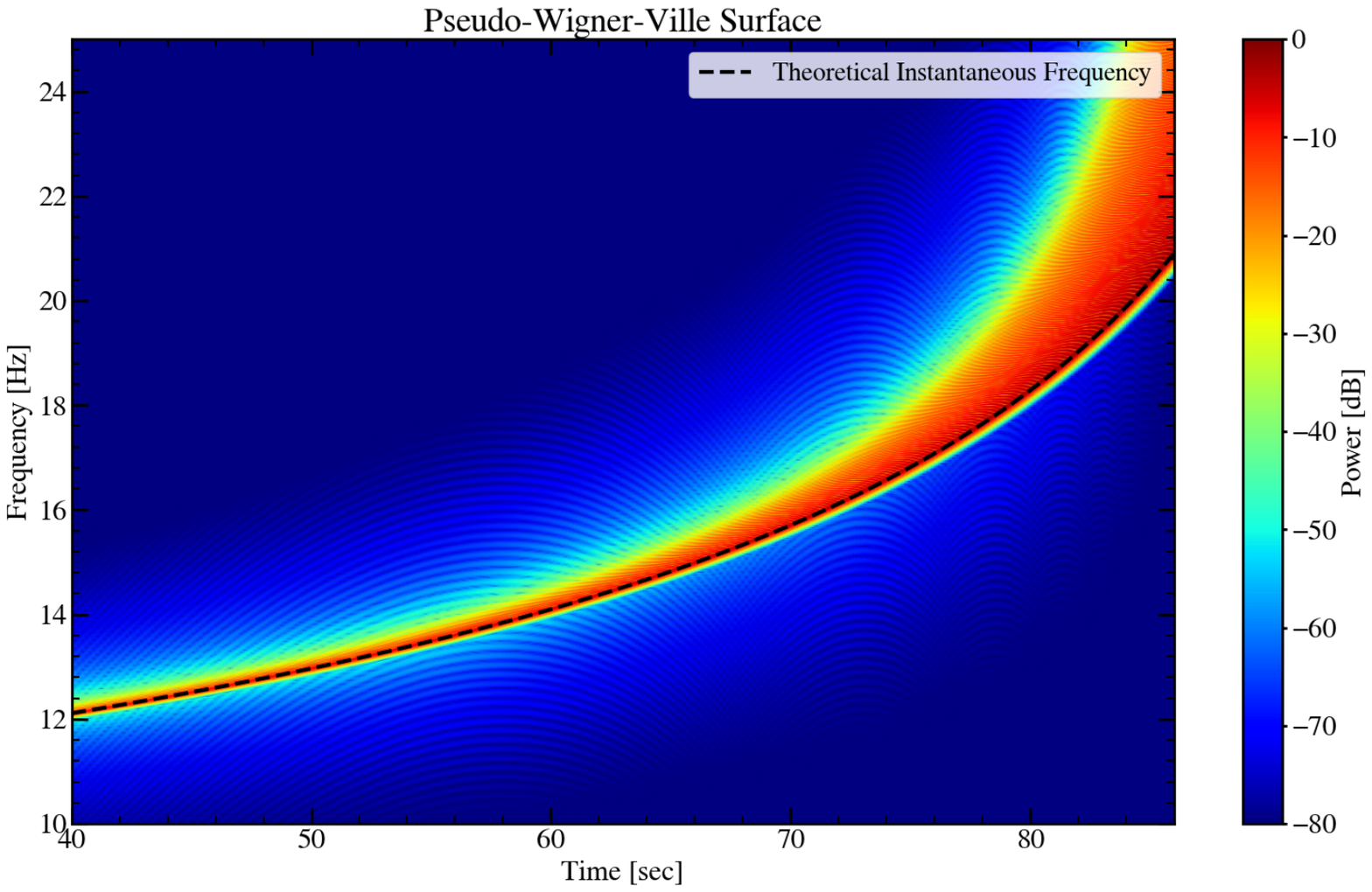}
    \caption{Time-frequency representation of the whitened gravitational wave signal using the PWVD. 
    The surface was computed numerically via the Fast Fourier Transform (FFT) with a Hanning window of $T_w = 50.0\, \mathrm{s}$ . 
    The chirp parameters are $f_0=10$ Hz, and $T_c=100\, \mathrm{s}$. 
    The $\nu_0=150$\, Hz  is the minimum noise frequency of the detector.
    The color axis represents the power in decibels (dB) relative to the global maximum, revealing the signal's energy concentration over time. 
    The ridge of the distribution closely follows the theoretical instantaneous frequency (black dashed line), 
    demonstrating the method's ability to track the signal's chirp evolution. The reddish band above the IF line is not a numerical artifact; it is the oscillating region of an incomplete Airy function, which attenuates exponentially below the IF.}
    \label{fig:pwvsurvace}
\end{figure}
The reddish tail above the instantaneous-frequency curve in Figure~\ref{fig:pwvsurvace} appears at first sight as an artifact, but it is in fact the oscillatory contribution predicted by the asymptotic Airy expansion.

Figure~\ref{fig:pwv_comparison} compares frequency slices of the PWVD at $t=50\,  \mathrm{s}$, 
computed with different approaches: direct numerical integration, FFT, uniform expansion, and the heuristic formula. 
Numerical methods and the uniform expansion agree closely, while the heuristic formula fails at higher frequencies.
Moreover, the frequency slice shows that the theoretical instantaneous frequency lies on the exponential rise of the Airy function, 
whereas the practical frequency estimate corresponds to the location of the Airy peak, as can be argued by Eq. (\ref{eq:BUAPWV}).
\begin{figure}[h!]
    \centering
    \includegraphics[width=1\textwidth]{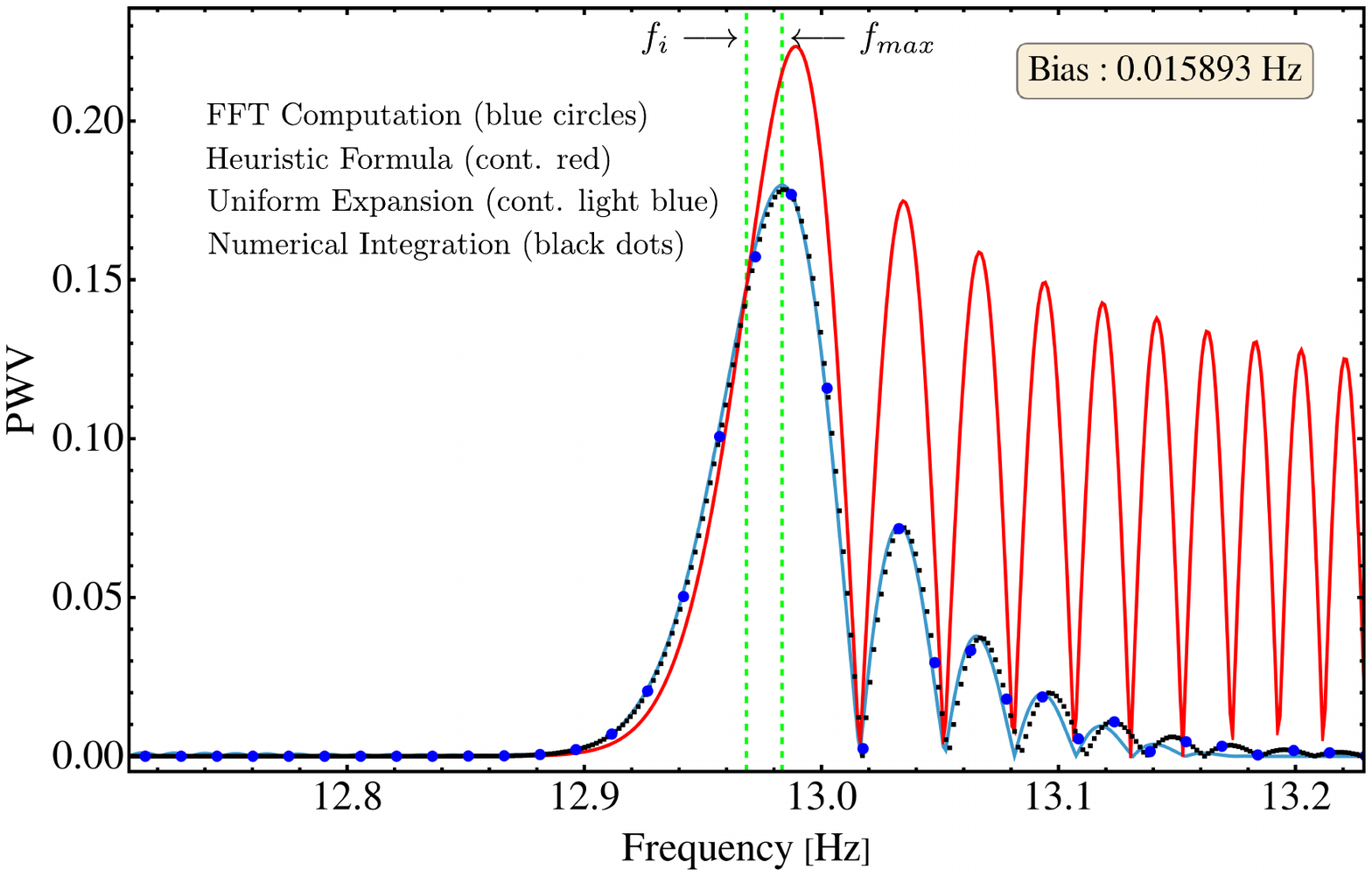}
    \caption{Comparison of the PWVD frequency slice ($t=50\,  \mathrm{s}$) of a whitened gravitational wave chirp as a function of frequency, computed using four different methods. 
    The results from direct numerical integration (black dots), FFT computation (blue circles), and the uniform expansion (continuous light-blue line) show excellent agreement across a broad region of the time-frequency plane. 
    The heuristic formula (continuous red line) accurately approximates the initial rise of the signal but fails to reproduce the correct oscillatory behavior at higher frequencies. The labels $f_i$ and $f_{max}$ indicates the instantaneous frequency and the asymptotic formula estimated peak respectively. 
    The chirp parameters are $f_0=10$ Hz, and $T_c=100\, \mathrm{s}$. We use $\nu_0=150$\, Hz, and $T_w = 50.0\, \mathrm{s}$.}
    \label{fig:pwv_comparison}
\end{figure}
We also evaluated the absolute ($\sim 0.0036$) and relative ($\sim 2.6$ \%) errors  between the numerical and asymptotic results in the 
vicinity of the PWVD peak. Both errors  remain small along the entire 
instantaneous–frequency line, confirming the robustness of the asymptotic description.

A systematic investigation of the bias in instantaneous frequency estimation by peak detection further supports the above conclusions. 
Figure~\ref{fig:wvestimations} shows the ridge-based frequency estimate compared to the theoretical curve, together with the resulting bias. 
The analysis reveals a small and time-dependent positive bias, whose magnitude is quantitatively consistent with the asymptotic predictions of eq. (\ref{eq:general_bias}).
\begin{figure}[h!]
    \centering
    \includegraphics[width=0.8\textwidth]{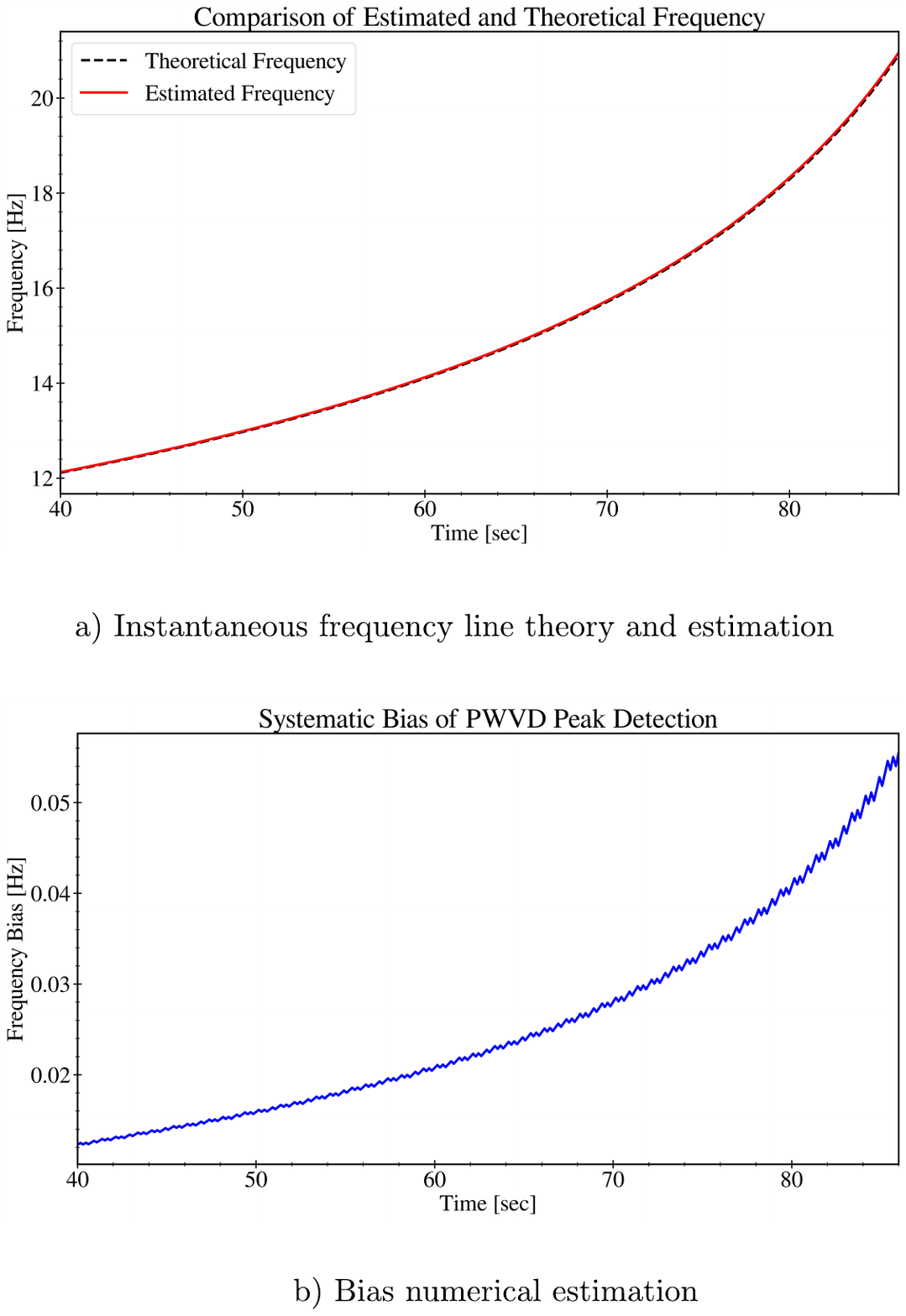}
    \caption{Analysis of the systematic bias in IF estimation from the Pseudo-Wigner-Ville (PWV) distribution. 
    The chirp parameters are $f_0=10$ Hz, and $T_c=100\, \mathrm{s}$.  We use $\nu_0=150$\, Hz, and $T_w = 50.0 \, \mathrm{s}$.
    \textbf{(a)} Comparison between the theoretical IF (black dashed line) and the IF estimated from the ridge of the PWVD surface (solid red line). 
    The peak at each time step was located with sub-sample precision using cubic spline interpolation to overcome the FFT grid resolution. 
    \textbf{(b)} The systematic bias, defined as the difference between the estimated and theoretical IF. 
    The result reveals a small, positive, and time-varying bias inherent to the PWV-based estimation method, consistent with theoretical predictions.}
    \label{fig:wvestimations}
\end{figure}
Table~\ref{tab:bias} summarizes the bias values obtained using different estimation methods, the numbers are relative to the frequency slice of Fig. \ref{fig:pwv_comparison}. 
The values derived from the asymptotic formula and from numerical simulations show close agreement, validating the theoretical framework.
When compared to the numerical results, the analytical predictions for the bias derived from Eqs. (\ref{eq:general_bias}) and (\ref{eq:general_bias_explicit}) exhibit a small systematic overestimation.
\begin{table}[h]
\centering
\caption{Comparison of bias values calculated with different methods. 
Relative to the slice $t=50\,  \mathrm{s}$, with $f_0=10$ Hz, $T_c=100\,  \mathrm{s}$, $T_w=50 \,  \mathrm{s}$, and $\nu_0=150$ Hz.}
\label{tab:bias}
\begin{tabular}{lc}
\hline
Calculation Method & Bias Value [Hz] \\
\hline
Numerical (Fig.~\ref{fig:pwv_comparison}) & 0.0159 \\
Equation (\ref{eq:BUAPWV}) & 0.0150 \\
Equation (\ref{eq:general_bias}) & 0.0209 \\
Equation (\ref{eq:general_bias_explicit}) & 0.0204 \\
\hline
\end{tabular}
\end{table}

In order to address the robustness of the uniform asymptotic expansion, we performed two additional validation tests. 
First, we repeated the comparison between the asymptotic expression and the numerical PWVD over several time slices along the chirp, rather 
than restricting the analysis to the representative slice at \(t=50\,\mathrm{s}\). 
The results, shown in Fig.~\ref{fig:time_robustness}, demonstrate that the agreement between the uniform expansion and the numerical 
computation persists over a broad portion of the chirp evolution. In particular, both the main lobe and the Airy-type 
oscillatory tail are reproduced with the same qualitative and quantitative accuracy as in the central slice.

\begin{figure}[t]
    \centering
    \includegraphics[width=0.95\linewidth]{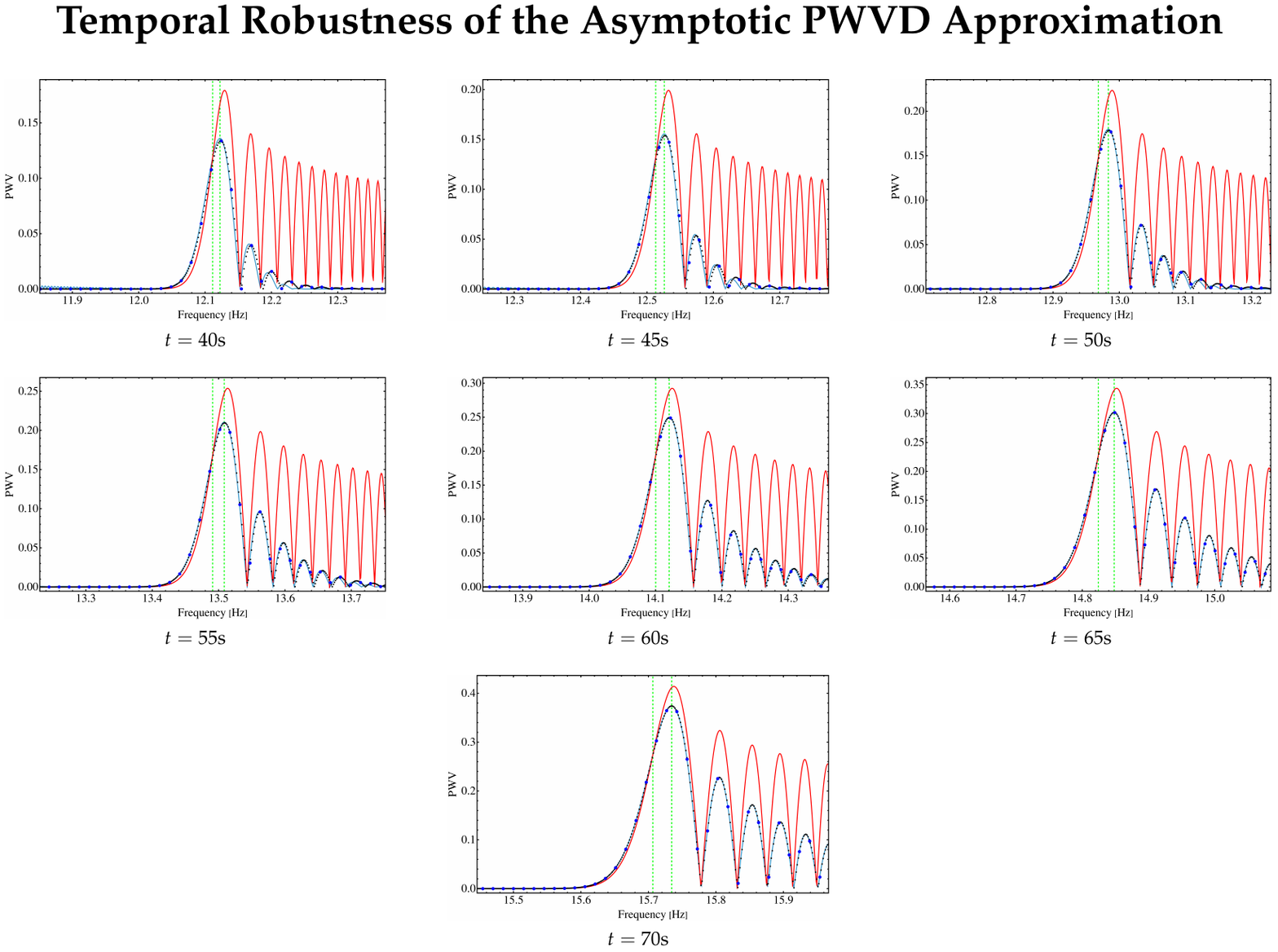}
    \caption{
    Time robustness of the uniform asymptotic expansion for the PWVD of the whitened gravitational-wave chirp. Frequency slices are shown at different analysis times along the chirp evolution. The numerical PWVD, computed by FFT, is compared with the corresponding asymptotic evaluation for the same signal and window parameters.
     The agreement remains accurate over the selected time interval, both in the position and amplitude of the dominant peak and in the first oscillatory lobes. 
     This demonstrates that the asymptotic approximation is not tuned to a single slice, but remains stable as the instantaneous frequency and the local curvature of the chirp evolve in time.
    }
    \label{fig:time_robustness}
\end{figure}

Second, we tested the stability of the comparison in the presence of additive broadband white Gaussian noise. 
The noise was injected directly in the real time-domain whitened strain, sampled at \(f_s=4096\,\mathrm{Hz}\), 
before constructing the analytic signal and computing the PWVD slice by FFT. 
The signal-to-noise ratio was defined as
\[
\mathrm{SNR}=\frac{\mathrm{rms}(h_w)}{\mathrm{rms}(n)},
\]
with the RMS evaluated on the physical signal interval. No band-limiting was applied to the injected noise, so that the noise occupies the full one-sided Nyquist band \(0\leq f\leq f_s/2\). 
As shown in Fig.~\ref{fig:noise_robustness}, the deterministic asymptotic profile remains consistent with the noisy FFT estimates at high SNR. 
The main deviations occur, as expected, in the low-amplitude oscillatory tail (more evident in the decibel representation), where the additive noise floor becomes dominant.

The noisy-signal tests confirm that the Airy-type features of the PWVD slice
persist under broadband Gaussian noise, including the asymmetric profile, the
oscillatory tail, the caustic decay, and the ridge bias. This supports the
robustness of the asymptotic description beyond the ideal noiseless case.

\begin{figure}[t!]
    \centering
    \includegraphics[width=0.8\linewidth]{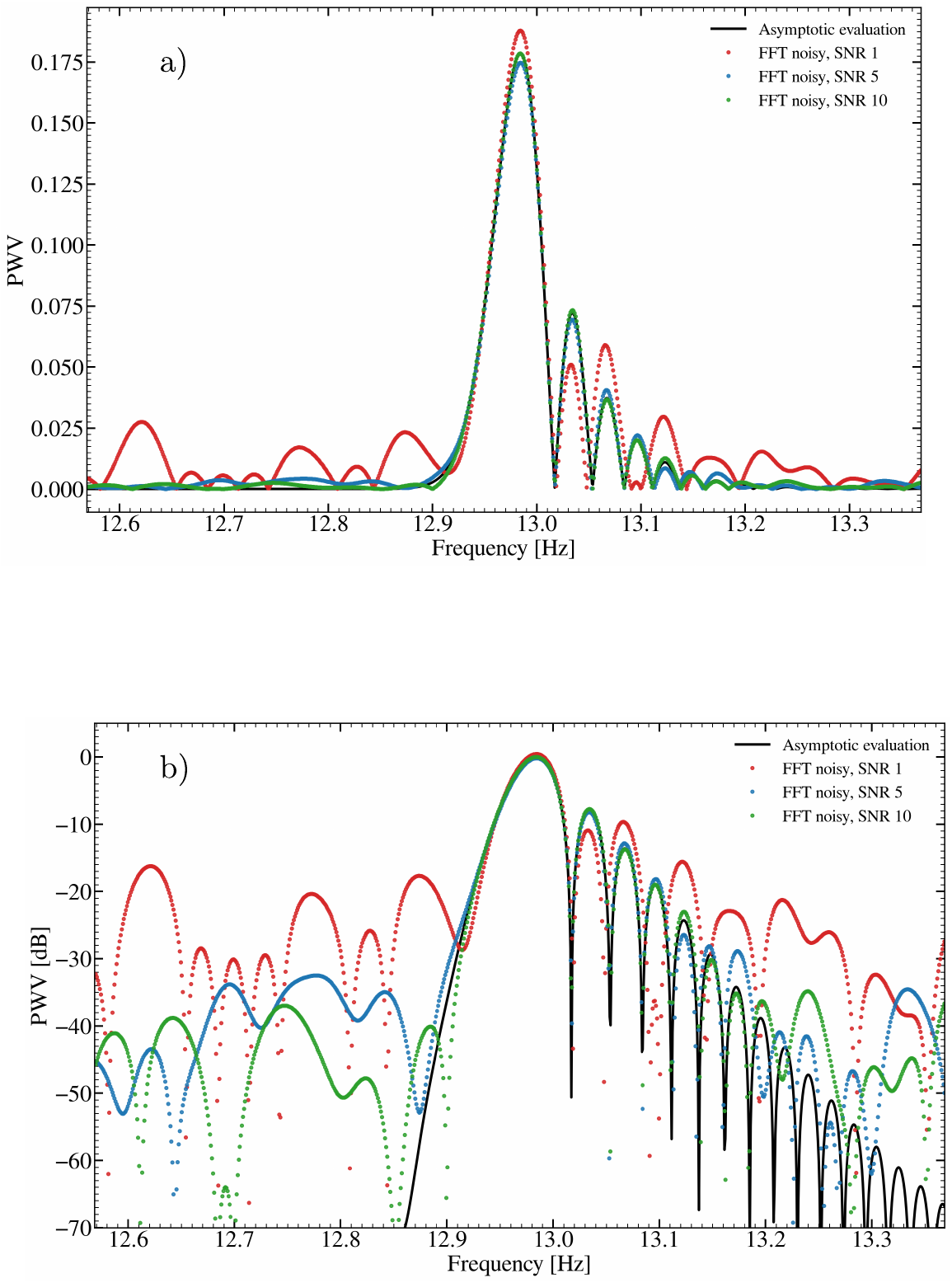}
    \caption{
    Robustness of the PWVD frequency slice at \(t=50\,\mathrm{s}\) in the presence of broadband additive white Gaussian noise. The black curve denotes the asymptotic evaluation, while the colored curves show FFT-based PWVD slices obtained after injecting white Gaussian noise into the time-domain whitened chirp for different RMS-based SNR values. 
    (a) Linear-scale representation. The dominant PWVD peak remains well localized and close to the asymptotic prediction, especially for increasing SNR. Noise mainly affects the lower-amplitude oscillatory structures away from the peak. 
    (b) Decibel-scale representation of the same slices. The logarithmic scale makes the broadband noise floor visible and shows that the Airy-like main structure is robust near the maximum, whereas the far oscillatory tail is progressively dominated by noise at lower SNR.
    }
    \label{fig:noise_robustness}
\end{figure}

\section{Application to Nonlinear Chirps for Radar Pulse-Compression}
\label{sec:ConcaveChirp}

Nonlinear chirps constitute a flexible class of frequency-modulated signals whose instantaneous frequency law can be adapted to the target application. 
In particular, concave nonlinear chirps have been considered in radar pulse-compression contexts because the curvature of the time-frequency trajectory can improve correlation properties 
and sidelobe behavior with respect to standard linear chirps ~\cite{Roy2021NLC}. 
Motivated by this observation, and in order to address the robustness of the proposed asymptotic expansion beyond a single convex chirp example, 
we also tested a representative class of concave chirps.

In the following of this section, we consider a nonlinear concave chirp whose instantaneous frequency is prescribed through a normalized time variable
\beq
u=\frac{t}{T_{\mathrm{chirp}}}, \qquad 0\leq u\leq 1,
\eeq
where $T_{\mathrm{chirp}}$ is the chirp duration.
The normalized instantaneous frequency is chosen as
\beq
\widehat f(u)
=
\alpha u
+\frac{1-\alpha}{2}u^{x_1}
+\frac{1-\alpha}{2}u^{x_2},
\eeq
with $\alpha=1.8$, $x_1=1.745$,$x_2=1.79$ (chosen according to the numerical example shown in~\cite{Roy2021NLC}).

The dimensional instantaneous frequency is then obtained as
\beq
f_i(t)=f_L+(f_H-f_L)\widehat f(t/T_{\mathrm{chirp}}).
\label{eq:concave}
\eeq
This construction satisfies \(\widehat f(0)=0\) and \(\widehat f(1)=1\), so that the chirp spans the prescribed frequency interval \([f_L,f_H]\). 
Since \(\alpha>1\) and \(x_1,x_2>1\), the second derivative of \(\widehat f\) is negative over the interval, producing a concave frequency law: 
the chirp rate is larger at early times and decreases progressively as the signal approaches the upper part of the band. 

This new model (\ref{eq:concave}) provides a complementary test with respect to the convex chirps considered above, 
because the curvature of the instantaneous frequency has the opposite sign.

For the numerical validation we use the parameter values $f_L=10~\mathrm{Hz}$, $f_H=500~\mathrm{Hz}$, $T_{\mathrm{chirp}}=100~\mathrm{s}$, and $T_w=50~\mathrm{s}$,
and analyze the PWVD  around \(t=50~\mathrm{s}\). 

Figure~\ref{fig:concave_pwvd_surface} shows the PWVD surface together with the theoretical instantaneous frequency. 
The ridge follows the expected chirp trajectory over the full analyzed time interval. 
The corresponding ridge bias, reported in Fig.~\ref{fig:concave_ridge_bias}, remains small and smooth over time, indicating that the observed frequency shift is not a local artifact of a single time slice. 
Compared with the gravitational chirp considered above, this new signal has a
less pronounced curvature over the analyzed time-frequency region. The
resulting PWVD ridge bias is correspondingly smaller, indicating that the
frequency displacement induced by the finite-time PWVD kernel is reduced when
the local chirp curvature is weaker.

Finally, Fig.~\ref{fig:concave_slice_t50} compares the PWVD slice at \(t=50\,\mathrm{s}\) with the heuristic approximation, the uniform asymptotic expansion, the FFT-based computation, and the direct numerical integration. 
The uniform expansion accurately reproduces the numerical reference across both the oscillatory and evanescent sides of the caustic,
 while the heuristic formula displays the expected loss of accuracy near the transition region.

\begin{figure}[!htbp]
    \centering
    \includegraphics[width=0.92\textwidth]{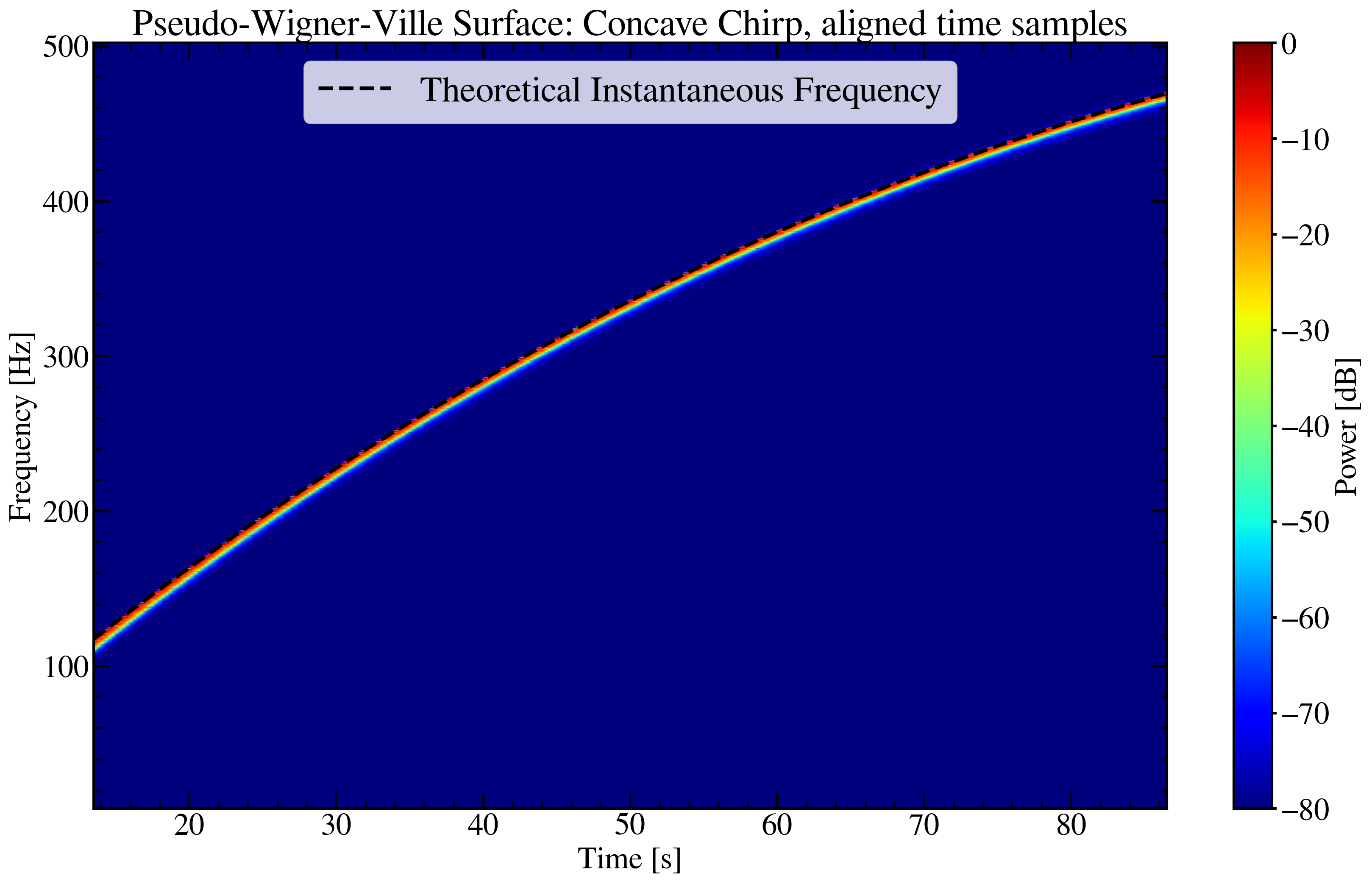}
    \caption{
    PWVD of the concave nonlinear chirp for the parameters $\alpha=1.8$, $x_1=1.745$,$x_2=1.79$, $f_L=10~\mathrm{Hz}$, $f_H=500~\mathrm{Hz}$, $T_{\mathrm{chirp}}=100~\mathrm{s}$ . 
    The dashed black curve denotes the theoretical instantaneous frequency. 
    The PWVD ridge follows the expected chirp trajectory over the full analyzed time interval, providing a global consistency check of the time-frequency representation.
    }
    \label{fig:concave_pwvd_surface}
\end{figure}

\begin{figure}[!htbp]
    \centering
    \includegraphics[width=0.88\textwidth]{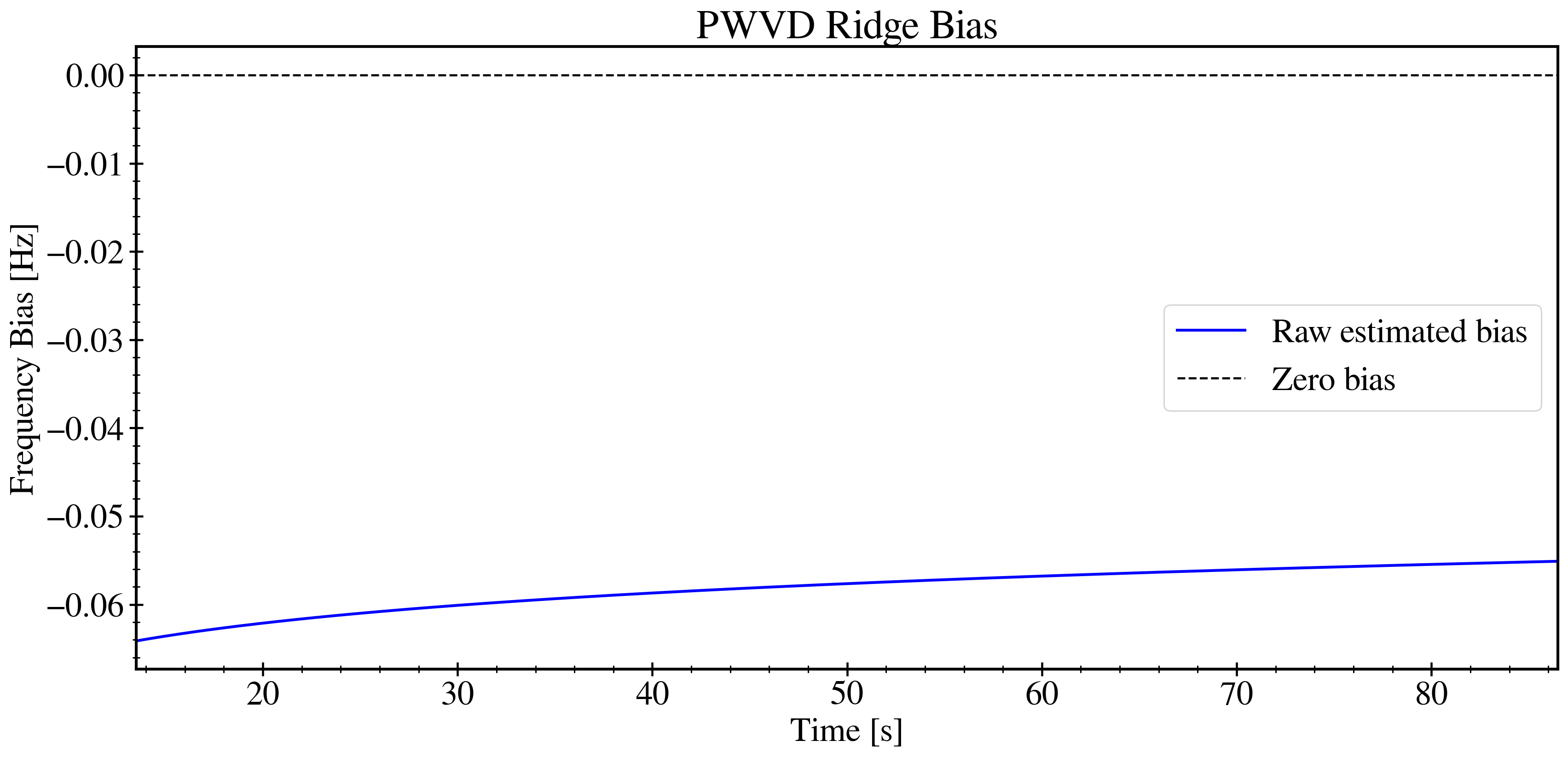}
    \caption{
    Estimated PWVD ridge bias for the concave nonlinear chirp as a function of time. 
    The chirp parameters are : $\alpha=1.8$, $x_1=1.745$,$x_2=1.79$, $f_L=10~\mathrm{Hz}$, $f_H=500~\mathrm{Hz}$, $T_{\mathrm{chirp}}=100~\mathrm{s}$ . 
    The bias is computed as the difference between the ridge frequency extracted from the PWVD and the theoretical instantaneous frequency. 
    Its smooth and small variation confirms that the frequency displacement is systematic and remains stable across the analyzed time interval.
    }
    \label{fig:concave_ridge_bias}
\end{figure}

\begin{figure}[!htbp]
    \centering
    \includegraphics[width=0.92\textwidth]{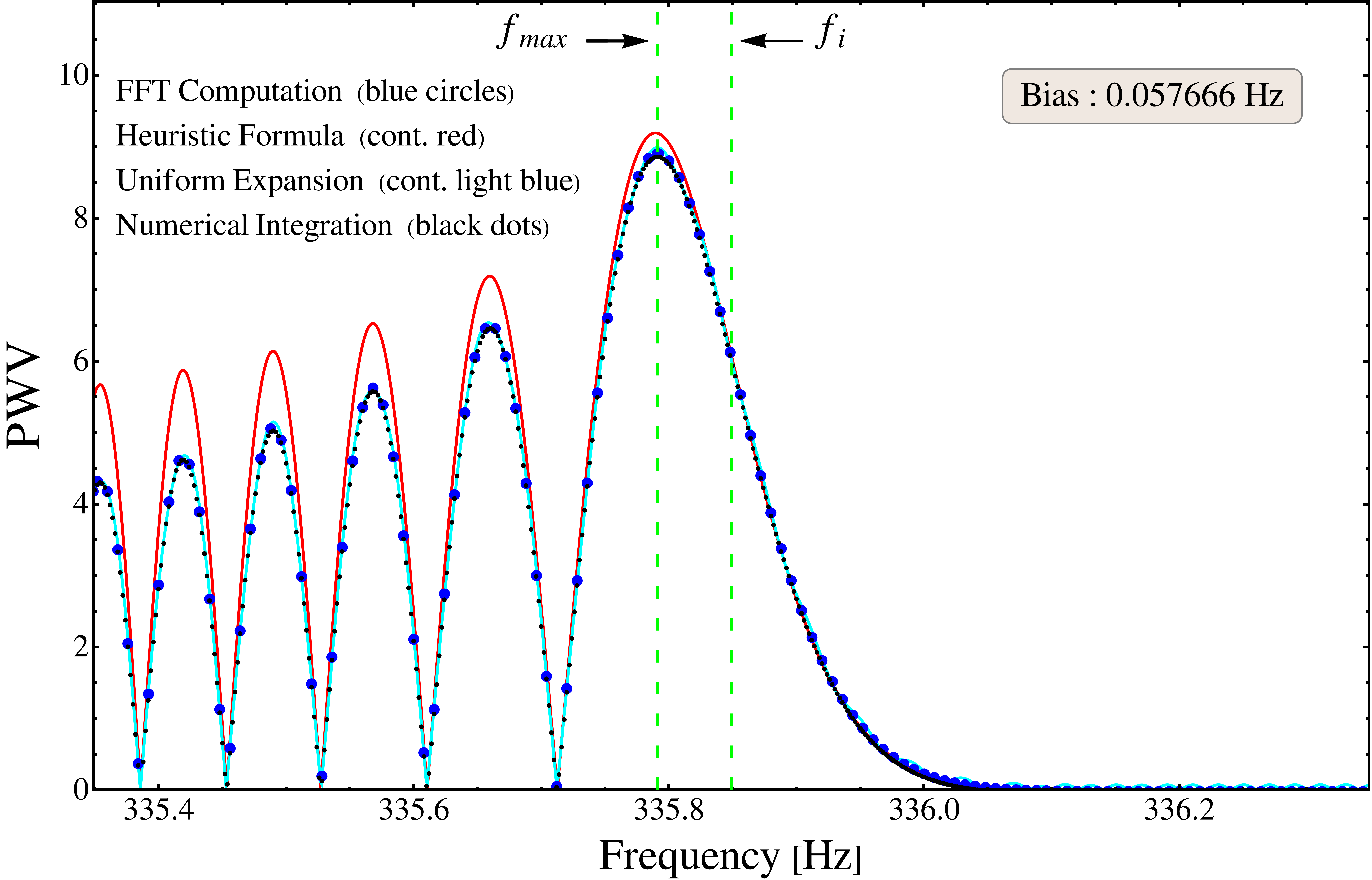}
    \caption{
    PWVD slice at \(t=50\,\mathrm{s}\) for the concave nonlinear chirp. 
    The chirp parameters are : $\alpha=1.8$, $x_1=1.745$,$x_2=1.79$, $f_L=10~\mathrm{Hz}$, $f_H=500~\mathrm{Hz}$, $T_{\mathrm{chirp}}=100~\mathrm{s}$ . 
    The FFT-based computation, direct numerical integration, heuristic formula, and uniform asymptotic expansion are compared on the same frequency interval. 
    The vertical dashed lines indicate the theoretical instantaneous frequency \(f_i\) and the frequency \(f_{\max}\) at which the PWVD ridge reaches its maximum. 
    The inset reports the corresponding ridge bias value.
    }
    \label{fig:concave_slice_t50}
\end{figure}

\section{Conclusions}
\label{sec:conclusions}

In this work, we have established a mathematically grounded framework to describe the interference structure of the Pseudo Wigner-Ville Distribution (PWVD) for nonlinear chirps.
Our core contribution lies in establishing a bridge between the classical uniform asymptotic theory of fold caustics \cite{chester1957, borovikov1994} 
and the applied domain of Time-Frequency signal processing. 
By adapting this mathematical machinery to the specific topological constraints of windowed distributions, we demonstrated that the PWVD is universally governed by 
symmetric incomplete Airy functions. 
This result offers a general mathematical key to decipher the complex phase-space interference patterns generated by quadratic time-frequency transforms. 
Indeed, it allows to reliably distinguish genuine physical features of the source from the transform-induced artifacts 
that have historically hindered the application of Wigner distributions in precision science~\cite{Mecklen2}.

Furthermore, our analysis clarifies and properly positions the connection between heuristic signal processing and geometric semiclassical mechanics \cite{Berry1977}. 
We demonstrated that while Berry's geometric chord construction provides an exact analogy for the \textit{phase} of the interference fringes, 
it is fundamentally incomplete at the \textit{amplitude} level, and should be analytically continued in the complex plane to be accurate in the evanescent region. 
By deploying the full uniform asymptotic machinery, we successfully incorporated the amplitude modulation and finite-window truncation effects that are inherently 
absent in infinite-domain quantum probability models.

The practical power of this framework was validated through the derivation of an analytical formula for the systematic bias inherent in peak-based instantaneous frequency estimation. 
The excellent agreement with numerical simulations of a gravitational-wave chirp { (with and without the presence of additive Gaussian noise) and of a radar nonlinear chirp} confirms that the PWVD, when supported by proper mathematical modeling of its nonlinearities, becomes a reliable quantitative tool. 
This demonstrates its readiness for demanding data analysis applications where precise signal characterization is paramount.

Looking forward, this work opens several avenues for research in nonlinear science:
\begin{enumerate}[(i)]
    \item \emph{Higher-order Caustics and Catastrophes}: From a mathematical perspective, our analysis addressed the {\it fold} catastrophe (two coalescing stationary points), 
    which is canonically mapped to the Airy function (see~\cite{ArnoldBook} for important results on catastrophe theory).
     A natural generalization involves chirps with higher-order phase degeneracies (e.g., inflection points in the frequency track) or multi component chirp signal.
      Investigating the PWVD in these regimes would require uniform expansions based on Pearcey functions (associated with the {\it cusp} catastrophe) 
      higher-order canonical integrals, {and uniform asymptotics for cross-term interference,} offering new insights into the topology of time-frequency surfaces.
    \item \emph{Advanced GW Data Analysis}: Within gravitational-wave science, this analytical tool can be integrated into search pipelines to provide rapid parameter pre-estimates and to develop robust, unmodeled search algorithms capable of detecting exotic physics beyond General Relativity.
    \item \emph{Dispersive and Nonlinear Media}: Finally, the proposed framework is applicable to other complex physical systems characterized by dispersive propagation,
     such as nonlinear optics, hydrodynamics, and plasma physics. In these fields, characterizing the fine structure of pulses via the PWVD could reveal subtle nonlinear 
     modulation effects often masked by standard spectrogram analysis~\cite{felsen1973radiation}.
\end{enumerate}

In conclusion, the methodology presented here offers a versatile instrument for any field dealing with nonlinear chirps in complex systems, 
paving the way for high-precision, non-parametric TF-based signal analysis.

\appendix

\section{Asymptotic Expansion of the Hilbert Transform}
\label{app:hilbert_theorem}

We present the theorem that provides a complete and rigorous derivation of the asymptotic formula for the Hilbert transform of an 
AM--FM signal (chirp) of the form $s_N(t)=a(t)\cos(N\varphi(t))$ as $N\to\infty$.

\begin{theorem}[Asymptotic Expansion of Analytic Signal]
\label{th:asinto_appendix}
Let $a \in C^2([0, T])$ and $\varphi \in C^3([0, T])$. We assume:
\begin{enumerate}
    \item \textbf{Non-degeneracy:}
    \[ \inf_{t \in [0, T]} |\varphi'(t)| := c > 0. \]
    \item \textbf{Bounded derivatives:} The functions and their derivatives $a, a', a''$ and $\varphi', \varphi'', \varphi'''$ are bounded on $[0, T]$. (This is guaranteed by the smoothness assumption on a closed interval).
\end{enumerate}
Let $\mathcal{H}$ be the Hilbert transform. Let the signal $s(t)$ be given by
\[ s(t) = a(t) \cos(N\varphi(t)). \]
Then, for any $t \in (0, T)$, the following asymptotic expansion holds for large $N$:
\[ \mathcal{H}\{s(t)\} = a(t)\sin(N\varphi(t)) - \frac{1}{N} \frac{d}{dt}\left(\frac{a(t)}{\varphi'(t)}\right) \cos(N\varphi(t)) + R_N(t), \]
where the remainder term $R_N(t)$ satisfies the uniform bound
\[ \sup_{t \in [0, T]} |R_N(t)| \leq C N^{-2}, \]
with the constant $C$ depending only on the norms $\|a\|_{W^{2,\infty}([0,T])}$, $\|\varphi\|_{W^{3,\infty}([0,T])}$, and the constant $c$.
\end{theorem}

\begin{remark}
Let us note that in the proof we decompose the signal as
\[
s_N(t) = \frac12 s_{+}(t) + \frac12 s_{-}(t), \qquad s_{\pm}(t) := a(t) e^{\pm iN\varphi(t)}.
\]
The $s_-$ term gives a contribution $\mathcal{O}(N^{-M})$ for every $M$ under the assumptions of Theorem~\ref{th:asinto_appendix}. If $a, \varphi$ are analytic in a strip in the complex plane, this contribution is exponentially small in $N$. The theorem can be rewritten under more general hypotheses and the expansion can be continued to higher orders.
\end{remark}

To make the statement of the theorem precise, we introduce in the formulation the Sobolev norms. For an interval $I \subset \mathbb{R}$, the Sobolev space $W^{k,p}(I)$ consists of functions $f$ such 
that $f$ and its weak derivatives up to order $k$ belong to the Lebesgue space $L^p(I)$.
For the case $p=\infty$, which is relevant here, the space $W^{k,\infty}(I)$ consists of functions whose weak derivatives up to order $k$ are essentially bounded. 
For functions that are continuously differentiable up to order $k$ (as assumed in our theorem), this simplifies. The norm $||\cdot||_{W^{k,\infty}(I)}$ measures the maximum absolute value achieved by the function or any of its derivatives up to order $k$.
These norms provide a compact way to quantify the {\it regularity} or {\it smoothness} of the functions $a$ and $\varphi$, which is essential for controlling the remainder term in the asymptotic expansion.

{\bf Proof:} This method provides a systematic way to derive the full asymptotic series to any order in $1/N$, where the regularity required for the functions $a(t)$ and $\varphi(t)$ increases with the order of the expansion. 
For a complete foundation of the theory, we refer the reader to the seminal work by Hörmander \cite{Hormander1985}. 
For modern treatments focused on the semiclassical perspective, see the books by Martinez \cite{Martinez2002} or Zworski \cite{Zworski2012}.
\qed

\section{Existence of Real Solutions for the Berry's Chord Equation}
\label{app:cordareale}

The nature of the solutions to the Berry's chord equation determines the local behavior of the WVD. 
A real solution $\tau_s$ corresponds to the {\it classically allowed} oscillatory region. 
While a global analysis depends on the overall structure of the instantaneous frequency $f_i(t)$, a 
\textit{local analysis} is sufficient to justify the uniform asymptotic framework. 
We prove here that for any point $t$ where $f_i(t)$ has non-zero curvature, there exists a frequency neighborhood around $f_i(t)$ where a real solution is guaranteed to exist.

\begin{theorem}[Local Existence of Real Solutions]
Let $f_i(t)$ be at least twice continuously differentiable at a time $t$. Let the midpoint function be $M(\tau_s) := \frac{1}{2}\left[f_i(t + \tau_s/2) + f_i(t - \tau_s/2)\right]$.
\begin{enumerate}
    \item If $f_i(t)$ is locally convex (i.e., $f_i''(t) > 0$), then there exists an $\epsilon > 0$ such that for any frequency $f \in [f_i(t), f_i(t) + \epsilon)$, the equation $M(\tau_s) = f$ has a real solution for $\tau_s$.
    \item If $f_i(t)$ is locally concave (i.e., $f_i''(t) < 0$), then there exists an $\epsilon > 0$ such that for any frequency $f \in (f_i(t) - \epsilon, f_i(t)]$, the equation $M(\tau_s) = f$ has a real solution for $\tau_s$.
\end{enumerate}
\label{th:cordareale}
\end{theorem}

\begin{proof}
The proof relies on a Taylor expansion of the midpoint function $M(\tau_s)$ for small values of $\tau_s$. We expand $f_i(t \pm \tau_s/2)$ around $t$:
\begin{align*}
    f_i\left(t + \frac{\tau_s}{2}\right) &= f_i(t) + f_i'(t)\frac{\tau_s}{2} + \frac{f_i''(t)}{2}\left(\frac{\tau_s}{2}\right)^2 + \mathcal{O}(\tau_s^3) \\
    f_i\left(t - \frac{\tau_s}{2}\right) &= f_i(t) - f_i'(t)\frac{\tau_s}{2} + \frac{f_i''(t)}{2}\left(\frac{\tau_s}{2}\right)^2 + \mathcal{O}(\tau_s^3)
\end{align*}
Substituting these into the definition of $M(\tau_s)$, the odd-powered terms in $\tau_s$ cancel out:
\[
    M(\tau_s) = \frac{1}{2} \left[ 2f_i(t) + 2\frac{f_i''(t)}{2}\frac{\tau_s^2}{4} + \mathcal{O}(\tau_s^4) \right] = f_i(t) + \frac{f_i''(t)}{8}\tau_s^2 + \mathcal{O}(\tau_s^4).
\]
For a frequency $f$ close to $f_i(t)$, we seek a small $\tau_s$ that solves $M(\tau_s) = f$. Using our approximation:
\[
    f \approx f_i(t) + \frac{f_i''(t)}{8}\tau_s^2 \implies \tau_s^2 \approx \frac{8(f - f_i(t))}{f_i''(t)}.
\]
For $\tau_s$ to be a real number, we require $\tau_s^2 \geq 0$. This leads to two cases:
\begin{itemize}
    \item \textbf{Case 1: Local Convexity ($f_i''(t) > 0$)}. For $\tau_s^2$ to be non-negative, the numerator must also be non-negative. This requires $f - f_i(t) \geq 0$, or $f \geq f_i(t)$.
    \item \textbf{Case 2: Local Concavity ($f_i''(t) < 0$)}. For $\tau_s^2$ to be non-negative, the numerator must be non-positive to cancel the negative sign of the denominator. This requires $f - f_i(t) \leq 0$, or $f \leq f_i(t)$.
\end{itemize}
Since this approximation is valid for $f$ sufficiently close to $f_i(t)$, the existence of a real solution $\tau_s$ is guaranteed in a local neighborhood on the side of the IF curve determined by its curvature. This completes the proof.
\end{proof}

\paragraph{Global Limitations.}
It is crucial to recognize that this local guarantee does not necessarily extend globally. The midpoint function $M(\tau_s)$ may be bounded, even if $f_i(t)$ is defined for all time. For instance, if $f_i(t)$ is defined on a finite time interval, the valid range for $\tau_s$ will also be finite, implying that $M(\tau_s)$ has a global maximum or minimum.
Consequently, there may exist a frequency $f_{bound}$ beyond which no real solution can be found, even if $f$ is in the correct {\it classically allowed} region. For example, in the convex case, the range of frequencies admitting real solutions is not necessarily $[f_i(t), \infty)$, but rather $[f_i(t), \max_{\tau_s} M(\tau_s)]$.
However, for the purpose of the uniform asymptotic analysis, which describes the WVD's structure \textit{near} the IF curve, the local existence theorem is the essential and sufficient condition.

\section{Bleistein Asymptotic Expansion}
\label{app:bleistein}

This appendix reports on the material present in the book by Borovikov \cite{borovikov1994}, p. 72. Consider an integral
\begin{equation}
I(\lambda, a, b) = \int_{b}^{\infty} f(x)g(x) \exp[i\lambda\varphi(x, a)] \,dx
\label{eq:integral}
\end{equation}
where for $a > 0$, the phase function $\varphi(x, a)$ has two stationary points, $x_1$ and $x_2$, close to the endpoint of the integration interval $x=b$.

The expansion is expressed in terms of the incomplete Airy function, defined by the integral:
\begin{equation}
\text{Ai}(p, q) = \frac{1}{2\pi} \int_{q}^{\infty} \exp\left\{i\left[\frac{t^3}{3} + tp\right]\right\} \,dt
\label{eq:incomplete_airy}
\end{equation}

To reduce the general integral \eqref{eq:integral} to the canonical model, we assume the function $\varphi(x, a)$ to be analytical for small $a$ and $x - x_0$. The following relations are required to hold at $x = x_0$, $a = 0$:
\begin{equation}
\label{eq:2.34}
\varphi'''_{xxx} \neq 0, \quad \varphi'_{x} = \varphi''_{xx} = 0, \quad \varphi''_{xa} \neq 0
\end{equation}
Under these conditions, there exists a change of variable $x = x(\tau)$ that is analytical and reversible for small $a$ and $x - x_0$, depending parametrically on $a$, which transforms $\varphi(x, a)$ into a cubic polynomial:
\begin{equation}
\label{eq:2.35}
\varphi(x, a) = \varphi_0(a) + \tau^3/3 - \xi(a)\tau
\end{equation}
where $\varphi_0(a)$ and $\xi(a)$ are analytical functions of $a$. The proof of this statement can be found in \cite{chester1957, levinson1960canonical}.

\begin{theorem}
Let $I(\lambda, a, b)$ be the integral from eq. (\ref{eq:integral}) in which for $a > 0$ and $a \to 0$ the phase function $\varphi(x, a)$
 has two merging stationary points $x_1$ and $x_2$, and the function $|\varphi''|$ in a vicinity of these points is bounded from below. 
 Let $g(x)$ be a cutting function and $f(x)$ be a function analytical in the region $\text{supp}(g)$. Then as $\lambda \to \infty$, 
 the asymptotic of $I$, uniform over the distance between the points $x_1$ and $x_2$ and the endpoint $b$ of the integration interval, is expressed in terms of the incomplete Airy function and its derivative as
\begin{equation}
\begin{split}
I(\lambda, a, b) = \frac{2\pi}{\lambda^{1/3}} \exp[i\lambda\varphi_0(a)] \bigg\{ & \mathcal{A} \Phi(-\lambda^{2/3}\xi, \lambda^{1/3}q) \\
& - \frac{i \mathcal{B}}{\lambda^{1/3}} \text{sign}(\varphi'') \Phi'_p(-\lambda^{2/3}\xi, \lambda^{1/3}q) \\
& + 2\pi i \frac{1}{\lambda^{2/3}} \text{sign}(\varphi'') \mathcal{C} \exp[i\lambda\varphi(b, a)] \bigg\}
\end{split}
\end{equation}
\label{th:uniforme}
\end{theorem}

Here, $\Phi(p, q)$ for $\varphi''>0$ coincides with the incomplete Airy function, $\text{Ai}(p,q)$, while for $\varphi''<0$ it coincides with its complex conjugate:
\begin{equation}
\Phi(p, q) = 
\begin{cases} 
\text{Ai}(p, q) & \text{for } \varphi'' > 0 \\
\text{Ai}^*(p, q) = \text{Ai}(p) - \text{Ai}(p, -q) & \text{for } \varphi'' < 0 
\end{cases}
\end{equation}
The functions $\varphi_0(a)$, $\xi = \xi(a)$, $q=q(a,b)$ and the coefficients of the asymptotic series $\mathcal{A} = \sum_{n=0}^{\infty} (i/\lambda)^n A_n$, $\mathcal{B} = \sum_{n=0}^{\infty} (i/\lambda)^n B_n$, $\mathcal{C} = \sum_{n=0}^{\infty} (i/\lambda)^n C_n$ are regular functions of the parameters $a, b$.

The functions $\varphi_0(a)$ and $\xi(a)$ are defined in terms of the original phase function at the stationary points $x_1$ and $x_2$:
\begin{gather}
\label{eq:phi0_xi_def}
\varphi_0(a) = \frac{\varphi(x_1, a) + \varphi(x_2, a)}{2} \\
\xi(a) = \left\{ \frac{3}{4} \left[ \varphi(x_1, a) - \varphi(x_2, a) \right] \right\}^{2/3} \nonumber
\end{gather}

The leading order coefficients, $A_0$ and $B_0$, of the asymptotic series are given by:
\begin{gather}
\label{eq:A0_B0_def}
A_0 = 2^{-1/2} \xi^{1/4} \left[ \frac{f(x_2)}{\sqrt{\varphi''_{xx}(x_2, a)}} + \frac{f(x_1)}{\sqrt{|\varphi''_{xx}(x_1, a)|}} \right] \\
B_0 = -2^{-1/2} \xi^{-1/4} \left[ \frac{f(x_1)}{\sqrt{|\varphi''_{xx}(x_1, a)|}} - \frac{f(x_2)}{\sqrt{\varphi''_{xx}(x_2, a)}} \right] \nonumber
\end{gather}

The function $q(a, b) = \tau(b)$, where $\tau(x)$ is the analytical and reversible variable change that transforms the phase function into a cubic polynomial eq. (\ref{eq:2.35}). For $a<0$, the function $q(a,b)$ is the unique root of the equation
\begin{equation}
\varphi(b, a) = \varphi_0(a) + \frac{q^3}{3} - \xi(a)q
\end{equation}
while for $a>0$ it can be defined as an analytical continuation from the region $a<0$.

Finally, the leading term $C_0$ of the asymptotic series $\mathcal{C}$ has the form
\begin{equation}
C_0 = \frac{f(b)}{\varphi'_b(a, b)} - \frac{A_0 - B_0 q}{q^2 - \xi}
\end{equation}

While Theorem~\ref{th:uniforme} addresses the case where the coalescing points are near a single integration endpoint, 
our work considers the scenario of a narrow and symmetric integration window. 
In this configuration, the stationary points are proximate to both endpoints, and the uniform expansion is generalized by
 employing symmetric incomplete Airy functions. 

\section{Uniform Asymptotic Expansion and Remainder Estimate}
\label{app:remainder_bound}

We present the uniform asymptotic expansion for the windowed PWVD into the following theorem, 
which adapts the general uniform stationary phase method \cite{borovikov1994, Bleistein1966} to the specific topological structure of the windowed PWVD.

\begin{theorem}[ Uniform Asymptotic Expansion of the PWVD]
\label{th:pwvd_uniform}
Let the PWVD integral be defined as in Eq. (\ref{eq:canonical}) with a large asymptotic parameter $\chi \gg 1$. Assume the following conditions hold:
\begin{enumerate}
    \item \textbf{Analyticity of the Signal:} The signal amplitude $a(t)$ and phase $\psi(t)$ are real-analytic functions on the analysis interval (as justified in Sec. \ref{sec:phase_analyticity}).
    \item \textbf{Tapered Window:} The window function $w(\tau)$ is real-analytic on the open interval $[-T/2, T/2]$ 
    and vanishes continuously at the boundaries of its support: $w(\pm T/2) = 0$. 
    (We note that standard tapered windows used in signal processing naturally satisfy this analyticity condition).
    \item \textbf{Non-degeneracy:} The instantaneous frequency curve has non-zero curvature at the analysis time $t$, i.e., $f_i''(t) \neq 0$.
\end{enumerate}
Then, the canonical mapping to the cubic phase exists as an analytic diffeomorphism, and the PWVD admits the following uniform asymptotic expansion:
\begin{equation}
PW_x(t, \bar{f}) = \left(\frac{2\pi}{\chi^{1/3}}\right) \left[ A_0 \Phi_{inc}(-\chi^{2/3}\xi, \chi^{1/3}q) \right] + R(\chi)
\label{eq:refined_uniform}
\end{equation}
Furthermore, the remainder term $R(\chi)$ is uniformly bounded by:
\begin{equation}
|R(\chi)| \le C(t, \bar{f}) \chi^{-1} = \order{\chi^{-1}}
\end{equation}
where $C(t, \bar{f})$ is a finite constant independent of $\chi$. 
Furthermore, the expansion is strictly uniform (i.e., $C(t, \bar{f})$ is bounded by a global constant $C_{max}$) 
over any compact region in the $(t, f)$ plane provided that two conditions are met: 
the instantaneous frequency curvature is strictly bounded away from zero ($|f_i''(t)| \ge c > 0$), 
and the stationary points remain strictly inside the window support ($\xi < q^2$).
\end{theorem}

\noindent\textbf{Remark:} Above \(\xi = \left[ \frac{3}{2} F(\tau_s, t, \bar{f}) \right]^{2/3}\) and \(q\) is defined by \(F(T, t, \bar{f}) = \frac{q^3}{3} - \xi q\).
Moreover, \(\xi\) does \emph{not} depend on the asymptotic parameter \(\chi\); it is defined solely from the phase \(F\). 
The parameter \(\chi\) enters only through the scaled arguments \(\chi^{2/3}\xi\) and \(\chi^{1/3}q\).

The proof of the remainder bound relies on the analyticity of the mapping and the vanishing of the amplitude at the endpoints. We provide the explicit derivation below.

\paragraph{Canonical mapping and transformed amplitude}
The uniform asymptotic expansion relies on the canonical mapping $\tau \mapsto u(\tau)$, defined by the phase relation $F(\tau, t, \bar{f}) = \frac{u^3}{3} - \xi u$. The transformed amplitude $G(u)$ is the product of the original PWVD amplitude and the Jacobian of the transformation:
\beq
G(u) = \underbrace{\left[ w\left(\frac{\tau(u)}{2}\right) w\left(-\frac{\tau(u)}{2}\right) a\left(t + \frac{\tau(u)}{2}\right) a\left(t - \frac{\tau(u)}{2}\right) \right]}_{A(\tau(u), t)} \cdot \underbrace{\left[ \frac{u^2 - \xi}{F'(\tau(u), t, \bar{f})} \right]}_{\frac{d\tau}{du}}
\label{eq:asiref1}
\eeq
Because the signal phase $\psi(t)$ is analytic (Assumption 1 of Theorem\,\, \ref{th:pwvd_uniform}), $F(\tau)$ is analytic. By Levinson's theorem \cite{levinson1960canonical}, the mapping $\tau(u)$ is an analytic diffeomorphism. The apparent singularity in the Jacobian at the stationary points (where $F'(\tau)=0$ and $u^2=\xi$) is strictly removable. Thus, $G(u)$ is analytic.

\paragraph{Symmetry and simplification of the Bleistein decomposition}
A fundamental property of the PWVD integral is its symmetry. The phase $F(\tau)$ is an odd function of $\tau$, which implies the mapping $u(\tau)$ is also odd. 
The amplitude $A(\tau,t)$ in Eq.\, \ref{eq:asiref1} is an even function of $\tau$, and the derivative $d\tau/du$ is even. Consequently, the transformed amplitude $G(u)$ is strictly an \emph{even} function: $G(-u) = G(u)$.

The standard Bleistein decomposition is $G(u) = A_0 + B_0 u + (u^2 - \xi)H(u)$. Evaluating this at the stationary points $u = \pm \sqrt{\xi}$ yields:
\beq
G(\sqrt{\xi}) = A_0 + B_0 \sqrt{\xi}, \quad G(-\sqrt{\xi}) = A_0 - B_0 \sqrt{\xi}
\eeq
Due to the even parity of $G(u)$, $G(\sqrt{\xi}) = G(-\sqrt{\xi})$, which strictly enforces $B_0 = 0$. This symmetry drastically simplifies the residual function $H(u)$ to:
\beq
H(u) = \frac{G(u) - A_0}{u^2 - \xi}
\eeq
Since $A_0 = G(\pm\sqrt{\xi})$, the numerator vanishes exactly at the roots of the denominator. 
Crucially, this regularity holds even at the exact coalescence point ($\xi=0$), where the denominator becomes $u^2$. 
In this limit, the even parity of $G(u)$ ensures that its first derivative vanishes at the origin ($G'(0)=0$). 
This provides a zero of order two in the numerator, perfectly matching the singularity of the denominator. 
By Riemann's theorem on removable singularities, all apparent singularities are strictly removable, making $H(u)$ globally analytic within the entire integration domain.

\paragraph{Evaluation of the remainder}
The exact remainder of the expansion is given by the integral of the $H(u)$ term over the transformed finite domain $[-q, q]$, where $\pm q$ are the images
of the original boundaries $\pm T$:
\beq
R(\chi) = \int_{-q}^{q} H(u) (u^2 - \xi) e^{i\chi(\frac{u^3}{3} - \xi u)} du
\eeq
Recognizing that $(u^2 - \xi) e^{i\chi(\frac{u^3}{3} - \xi u)} = \frac{1}{i\chi} \frac{d}{du} \left( e^{i\chi(\frac{u^3}{3} - \xi u)} \right)$, we apply integration by parts:
\beq
R(\chi) = \frac{1}{i\chi} \left[ H(u) e^{i\chi(\frac{u^3}{3} - \xi u)} \right]_{-q}^{q} - \frac{1}{i\chi} \int_{-q}^{q} H'(u) e^{i\chi(\frac{u^3}{3} - \xi u)} du
\eeq
We can explicitly evaluate the boundary term. Since $G(u)$ is an even function, $H(u)$ is also strictly even, meaning $H(-q) = H(q)$. Furthermore, the phase function is odd. This symmetry allows us to combine the complex exponentials into a sine function:
\beq
R(\chi) = \frac{2 H(q)}{\chi} \sin\left[\chi\left(\frac{q^3}{3} - \xi q\right)\right] - \frac{1}{i\chi} \int_{-q}^{q} H'(u) e^{i\chi(\frac{u^3}{3} - \xi u)} du
\label{eq:reminder}
\eeq
Let us underline that the explicit endpoint contribution at order \(O(\chi^{-1})\) is controlled by the boundary value \(H(q)\) in Eq. (\ref{eq:reminder}). 
The second term in Eq. (\ref{eq:reminder}) involving  \(H'(u)\) should instead be understood as the residual oscillatory integral that generates 
higher-order corrections upon further integrations by parts (following the general methods \cite{ Bleistein1966}). 

Taking the absolute value, and using the triangle inequality along with the fact that $|\sin(\cdot)| \le 1$, we obtain the strict bound:
\beq
|R(\chi)| \le \frac{1}{\chi} \left( 2|H(q)| + \int_{-q}^{q} |H'(u)| du \right)
\label{eq:remainder_bound_ineq}
\eeq
Crucially, because the window function $w(\tau)$ vanishes at the boundaries $\tau = \pm T/2$ (Assumption 2 of Thm. \ref{th:pwvd_uniform}), the transformed amplitude vanishes at the endpoints: $G(\pm q) = 0$. 
Therefore, the boundary evaluation of the residual function simplifies to a constant independent of $\chi$: 
\beq
H(q) = \frac{-A_0}{q^2 - \xi}.
\eeq
Provided the stationary points remain strictly inside the window (i.e., $\xi < q^2$, meaning the caustics do not hit the boundary of the support), this boundary term is finite and well-defined.
This explicit form of $H(q)$ reveals a crucial property regarding the window size $T$. 

The parameter $q$ is defined by the mapping evaluated at the boundary: $F(T, t, \bar{f}) = q^3/3 - \xi q$.
 As the window size $T$ increases, the transformed boundary $q$ increases accordingly. 
 Consequently, the magnitude of the boundary term $H(q)$ decays as $\mathcal{O}(q^{-2})$. 

This shows that the explicit endpoint contribution at order $\order{\chi^{-1}}$ is a finite-window effect.
In the limit of an infinitely wide window ($T \to \infty \implies q \to \infty$), the boundary term $H(q) \to 0$, and this explicit endpoint contribution vanishes.
The remaining truncation error is then governed by the intrinsic higher-order terms of the uniform Airy expansion, 
whose first neglected contribution is of order $\order{\chi^{-4/3}}$ under the assumptions of Theorem~\ref{th:pwvd_uniform}.

In conclusion, since $H(u)$ is analytic on the compact interval $[-q, q]$, its derivative $H'(u)$ is continuous and bounded. 
The integral $\int_{-q}^{q} |H'(u)| du$ is therefore strictly finite. 
This proves that both the boundary term and the integral term in Eq. (\ref{eq:remainder_bound_ineq}) scale exactly as $\order{\chi^{-1}}$. Therefore, the total remainder is uniformly bounded by $\order{\chi^{-1}}$, which is asymptotically subdominant to the leading $\order{\chi^{-1/3}}$ incomplete Airy term. 
This provides the required bound for the remainder within the stated assumptions of Theorem~\ref{th:pwvd_uniform}.

\paragraph{Calculation of the leading coefficient $A_0$}
The coefficient $A_0$ is determined by evaluating the transformed amplitude $G(u)$ at the stationary points $u = \pm \sqrt{\xi}$.
 At these points, the derivative of the phase $F'(\tau)$ vanishes, leading to an indeterminate form $0/0$ for the Jacobian $d\tau/du$. 
 Applying L'Hôpital's rule (or equivalently, differentiating the mapping relation $F(\tau) = u^3/3 - \xi u$ twice with respect to $u$), we obtain:
\beq
F''(\tau_s) \left(\frac{d\tau}{du}\right)^2 + F'(\tau_s) \frac{d^2\tau}{du^2} = 2u \implies \left(\frac{d\tau}{du}\right)_{u=\pm\sqrt{\xi}} = \left| \frac{2\sqrt{\xi}}{F''(\tau_s, t, \bar{f})} \right|^{1/2}
\eeq
Substituting this Jacobian into Eq. (\ref{eq:asiref1}) evaluated at the stationary points yields the explicit formula for $A_0$:
\beq
A_0 = G(\pm\sqrt{\xi}) = \xi^{1/4} A(\tau_s, t) \left| \frac{2}{F''(\tau_s, t, \bar{f})} \right|^{1/2}
\eeq
which matches the expression provided in the main text.

\section{Adiabatic Filtering Theorem}
\label{app:adiabfilt}
This appendix provides a formulation and proof of the adiabatic approximation for filtering a chirp signal. 
We compute the asymptotic regime using a large parameter $\lambda$ and by meticulously executing a two-stage application of Stationary Phase Method  (SPM). 
The result shows that for a slowly varying chirp, the output of a Linear Time-Invariant (LTI) filter is approximately the input signal's amplitude multiplied by the filter's transfer function evaluated at the instantaneous frequency.

\begin{theorem}{Adiabatic Filtering}
Consider a complex analytic chirp signal $z(t)$ defined as:
\[
z(t) = A(t) e^{i\lambda\psi(t)}
\]
where:
\begin{enumerate}
    \item $\lambda \gg 1$ is a large, dimensionless asymptotic parameter.
    \item The amplitude $A(t) \in C^1(\R)$ is a "slowly-varying" function, meaning $A'(t) = \order{1}$ with respect to $\lambda$.
    \item The phase function $\psi(t) \in C^3(\R)$ is also "slowly-varying," with derivatives $\psi'(t), \psi''(t), \psi'''(t) = \order{1}$.
    \item The phase is non-stationary, i.e., $\psi''(t) \neq 0$ for all $t$ in the domain of interest.
\end{enumerate}

The instantaneous frequency of the signal is given by:
\[
f(t) = \frac{1}{2\pi}\frac{d}{dt}\left(\lambda\psi(t)\right) = \frac{\lambda}{2\pi}\psi'(t)
\]

Let $S$ be a Linear Time-Invariant (LTI) system with a frequency response $H(f) \in C^1(\R)$ that has compact support (or is rapidly decaying).\\
{\bf Claim}
The filtered output $y(t) = (S*z)(t)$ satisfies the approximation:
\[
y(t) = A(t)H(f(t))e^{i\lambda\psi(t)} + R_\lambda(t)
\]
where the residual term $R_\lambda(t)$ is uniformly bounded by:
\[
\|R_\lambda(t)\|_\infty = \order{\lambda^{-1}}
\]

\label{th:adiabfilt}
\end{theorem}

\begin{proof}
The proof proceeds by applying the Method of Stationary Phase twice.

\subsubsection*{Step 1: Fourier Transform of the Chirp Signal}
The filtered output in the time domain is the inverse Fourier transform of the product of the filter's response and the signal's spectrum, $\hat{z}(f)$. First, we compute $\hat{z}(f)$:
\[
\hat{z}(f) = \int_{-\infty}^{\infty} z(\tau) e^{-i2\pi f \tau} d\tau = \int_{-\infty}^{\infty} A(\tau) e^{i(\lambda\psi(\tau) - 2\pi f \tau)} d\tau
\]
This is an oscillatory integral with a phase $\Phi_f(\tau) = \lambda\psi(\tau) - 2\pi f \tau$. The stationary point $\tau_f$ is found by setting the first derivative of the phase to zero:
\[
\frac{d\Phi_f}{d\tau}(\tau_f) = \lambda\psi'(\tau_f) - 2\pi f = 0 \implies \lambda\psi'(\tau_f) = 2\pi f
\]
This condition implicitly defines the stationary time point $\tau_f$ as a function of frequency $f$. The second derivative of the phase is $\Phi_f''(\tau) = \lambda\psi''(\tau)$.

\subsubsection*{Step 2: First Application of the Stationary Phase Method (SPM)}
Applying the SPM formula to the integral for $\hat{z}(f)$, we obtain the leading-order approximation:
\[
\hat{z}(f) \approx A(\tau_f) \sqrt{\frac{2\pi}{|\Phi_f''(\tau_f)|}} e^{i\Phi_f(\tau_f) + i\frac{\pi}{4}\sigma_f} = A(\tau_f) \sqrt{\frac{2\pi}{\lambda|\psi''(\tau_f)|}} e^{i(\lambda\psi(\tau_f) - 2\pi f \tau_f) + i\frac{\pi}{4}\sigma_f}
\]
where $\sigma_f = \text{sign}(\psi''(\tau_f))$. The error in this approximation is $\order{\lambda^{-1}}$.

\subsubsection*{Step 3: Filtering in the Frequency Domain}
The output signal $y(t)$ is the inverse Fourier transform of $H(f)\hat{z}(f)$:
\[
y(t) = \int_{-\infty}^{\infty} H(f)\hat{z}(f) e^{\mathrm{i} 2\pi f t} df
\]
Substituting our approximation for $\hat{z}(f)$ yields another oscillatory integral:
\[
y(t) \approx \int_{-\infty}^{\infty} G(f) e^{\mathrm{i} \Psi(f)} df
\]
where the new amplitude $G(f)$ and phase $\Psi(f)$ are:
\begin{align*}
G(f) &= H(f) A(\tau_f) \sqrt{\frac{2\pi}{\lambda|\psi''(\tau_f)|}} e^{i\frac{\pi}{4}\sigma_f} \\
\Psi(f) &= \lambda\psi(\tau_f) - 2\pi f \tau_f + 2\pi f t
\end{align*}

\subsubsection*{Step 4: Second Application of the SPM}
We find the stationary point $f_t$ of the integral for $y(t)$ by differentiating $\Psi(f)$:
\[
\frac{d\Psi}{df} = \underbrace{(\lambda\psi'(\tau_f) - 2\pi f)}_{=0 \text{ by definition of } \tau_f} \frac{d\tau_f}{df} - 2\pi\tau_f + 2\pi t
\]
The stationary condition $\frac{d\Psi}{df} = 0$ thus simplifies to $\tau_f = t$. This implies the stationary point is the instantaneous frequency $f_t = f(t)$. The second derivative is $\Psi''(f(t)) = -\frac{4\pi^2}{\lambda\psi''(t)}$.

\subsubsection*{Step 5: Simplification and Final Result}
Applying the SPM formula to $y(t)$ at the stationary point $f(t)$:
\[
y(t) \approx G(f(t)) \sqrt{\frac{2\pi}{|\Psi''(f(t))|}} e^{i\Psi(f(t)) + i\frac{\pi}{4}\theta}
\]
where $\theta = \text{sign}(\Psi''(f(t)))$. As shown in the detailed derivation, the amplitude terms simplify to $A(t)H(f(t))$, the phase $\Psi(f(t))$ becomes $\lambda\psi(t)$, and the complex phase factors $e^{i\pi/4}$ cancel out, leading to:
\[
y(t) \approx A(t) H(f(t)) e^{\mathrm{i} \lambda\psi(t)}
\]

\subsubsection*{Step 6: Error Analysis}
Each application of the SPM introduces an error of order $\order{\lambda^{-1}}$. 
The total residual $R_\lambda(t)$ is therefore the sum of two terms of order $\order{\lambda^{-1}}$, resulting in an overall error of $\order{\lambda^{-1}}$.
\end{proof}

\paragraph{Additional details on the Error Analysis}
The main proof provides the essential steps for deriving the leading-order term.
 A crucial question is how the error from the first application of SPM affects the final error bound. 
 This appendix provides a more detailed look at this error propagation.

\paragraph{Decomposition of the Signal Spectrum.}
Let us write the Fourier transform $\hat{z}(f)$ from Step 1 exactly as its SPM approximation plus a residual error term, $E_1(f)$:
\[
\hat{z}(f) = \hat{z}_{\text{approx}}(f) + E_1(f)
\]
From the theory of the Method of Stationary Phase, we know that this error is uniformly bounded:
\[
\|E_1(f)\|_\infty = \order{\lambda^{-1}}
\]

\paragraph{Decomposition of the Output Signal.}
Now, substitute this exact expression into the formula for the output $y(t)$:
\[
y(t) = \int_{-\infty}^{\infty} H(f) \left[ \hat{z}_{\text{approx}}(f) + E_1(f) \right] e^{\mathrm{i} 2\pi f t} df
\]
By linearity of the integral, we can split this into two parts:
\[
y(t) = \underbrace{\int_{-\infty}^{\infty} H(f) \hat{z}_{\text{approx}}(f) e^{i2\pi f t} df}_{\text{Term 1}} + \underbrace{\int_{-\infty}^{\infty} H(f) E_1(f) e^{i2\pi f t} df}_{\text{Term 2 (Propagated Error)}}
\]

\paragraph{Analysis of Term 1.}
Term 1 is precisely the integral we analyzed in the main proof (Steps 4 and 5). Applying the SPM to it yields our final approximation plus a new error term, $E_2(t)$:
\[
\text{Term 1} = A(t)H(f(t))e^{\mathrm{i} \lambda\psi(t)} + E_2(t)
\]
where $\|E_2(t)\|_\infty = \order{\lambda^{-1}}$.

\paragraph{Analysis of Term 2 (Propagated Error).}
Term 2, let's call it $E_{\text{prop}}(t)$, is the inverse Fourier transform of the product of the filter response and the first-stage error. We need to bound its magnitude:
\[
|E_{\text{prop}}(t)| = \left| \int_{-\infty}^{\infty} H(f) E_1(f) e^{\mathrm{i} 2\pi f t} df \right|
\]
Using the triangle inequality for integrals and the fact that $|e^{\mathrm{i} 2\pi f t}| = 1$:
\[
|E_{\text{prop}}(t)| \le \int_{-\infty}^{\infty} |H(f)| |E_1(f)| df
\]
We know that $|E_1(f)| \le C_1 \lambda^{-1}$ for some constant $C_1$. Furthermore, by hypothesis, $H(f)$ has a compact support (or is rapidly decaying and integrable) and is bounded, $|H(f)| \le M_H$. Therefore, we can bound the integral:
\begin{align*}
|E_{\text{prop}}(t)| &\le \int_{\supp(H)} M_H \cdot (C_1 \lambda^{-1}) df \\
&\le M_H C_1 \lambda^{-1} \int_{\supp(H)} 1 \, df
\end{align*}
The remaining integral, $\int_{\supp(H)} 1 \, df$, is simply the measure (or effective bandwidth) of the support of $H(f)$, which is a finite constant, let's call it $L_H$. This gives the final bound:
\[
|E_{\text{prop}}(t)| \le (M_H C_1 L_H) \lambda^{-1} \implies \|E_{\text{prop}}(t)\|_\infty = \order{\lambda^{-1}}
\]

\paragraph{Conclusion on Total Error.}
The exact output $y(t)$ is the sum of our findings:
\[
y(t) = \left( A(t)H(f(t))e^{i\lambda\psi(t)} + E_2(t) \right) + E_{\text{prop}}(t)
\]
The total residual $R_\lambda(t)$ is the sum of the two error terms:
\[
R_\lambda(t) = E_2(t) + E_{\text{prop}}(t)
\]
Since both $\|E_2(t)\|_\infty$ and $\|E_{\text{prop}}(t)\|_\infty$ are of order $\order{\lambda^{-1}}$, their sum is also of order $\order{\lambda^{-1}}$. 
This confirms that the error from the first stage does not increase the overall order of error in the final result.

\end{document}